\documentclass[pdflatex, sn-mathphys, Numbered]{sn-jnl}% Math and Physical Sciences Reference Style
\usepackage{graphicx}%
\usepackage{multirow}%
\usepackage{amsmath,amssymb,amsfonts}%
\usepackage{amsthm}%
\usepackage{mathrsfs}%
\usepackage[title]{appendix}%
\usepackage{xcolor}%
\usepackage{textcomp}%
\usepackage{manyfoot}%
\usepackage{booktabs}%
\usepackage{algorithm}%
\usepackage{algpseudocode}%
\usepackage{listings}%

\usepackage{eurosym}
\usepackage{float}
\usepackage[tight,footnotesize]{subfigure}
\usepackage{rotating}
\usepackage{enumerate}
\usepackage{color,soul}
\usepackage[colorinlistoftodos]{todonotes}
\usepackage{lipsum}
\usepackage{cases}
\usepackage{comment}
\usepackage{resizegather}
\usepackage{xurl}
\usepackage{tablefootnote}
\usepackage{array}
\newcolumntype{C}[1]{>{\centering\arraybackslash$}p{#1}<{$}}

\usepackage{xpatch}
\usepackage{nicematrix}

\xpatchcmd{\proof}{\itshape}{\prooflabelfont}{}{}
\newcommand{\prooflabelfont}{\bfseries}

\newcommand{\bs}{\boldsymbol{\mathrm{s}}}
\newcommand{\br}{\boldsymbol{\mathrm{r}}}
\newcommand{\be}{\boldsymbol{\mathrm{e}}}

\newcommand{\bu}{\boldsymbol{\mathrm{u}}}
\newcommand{\bt}{\boldsymbol{\mathrm{t}}}
\newcommand{\bA}{\boldsymbol{\mathrm{A}}}
\newcommand{\buone}{\bA^{T}\br+\be_{1}}

\newcommand{\coloneqq}{:=}

\DeclareMathOperator{\var}{Var}
\DeclareMathOperator{\diag}{diag}

\DeclareMathSizes{10}{9}{7}{5}
\newtheorem{theo}{Theorem}

\newtheorem{theorem}{Theorem}

\newtheorem{definition}[theorem]{Definition}

\newtheorem{lemma}{Lemma}

\newtheorem{remark}{Remark}

\newtheorem{example}{Example}

\begin{document}

\title[Article Title]{Lattice Codes for CRYSTALS-Kyber}

%%=============================================================%%
%% Prefix	-> \pfx{Dr}
%% GivenName	-> \fnm{Joergen W.}
%% Particle	-> \spfx{van der} -> surname prefix
%% FamilyName	-> \sur{Ploeg}
%% Suffix	-> \sfx{IV}
%% NatureName	-> \tanm{Poet Laureate} -> Title after name
%% Degrees	-> \dgr{MSc, PhD}
%% \author*[1,2]{\pfx{Dr} \fnm{Joergen W.} \spfx{van der} \sur{Ploeg} \sfx{IV} \tanm{Poet Laureate} 
%%                 \dgr{MSc, PhD}}\email{iauthor@gmail.com}
%%=============================================================%%

\author*[1]{\fnm{Shuiyin} \sur{Liu}}\email{SLiu@Holmes.edu.au}

\author[2]{\fnm{Amin} \sur{Sakzad}}\email{Amin.Sakzad@monash.edu}
%\equalcont{These authors contributed equally to this work.}

%\author[1,2]{\fnm{Third} \sur{Author}}\email{iiiauthor@gmail.com}
%\equalcont{These authors contributed equally to this work.}

\affil*[1]{\orgdiv{Cyber Security Research and Innovation Centre}, \orgname{Holmes
Institute}, \orgaddress{\city{Melbourne}, \postcode{VIC 3000}, \state{Victoria}, \country{Australia}}}

\affil[2]{\orgdiv{Department of Software Systems $\&$ Cybersecurity}, \orgname{Monash University}, \orgaddress{\city{Melbourne}, \postcode{VIC 3800}, \state{Victoria}, \country{Australia}}}

%\affil[3]{\orgdiv{Department}, \orgname{Organization}, \orgaddress{\street{Street}, \city{City}, %\postcode{610101}, \state{State}, \country{Country}}}

%%==================================%%
%% sample for unstructured abstract %%
%%==================================%%

\abstract{This paper describes a constant-time lattice encoder for the National Institute of Standards and Technology (NIST) recommended post-quantum encryption algorithm: Kyber. The first main contribution of this paper is to refine the analysis of Kyber decoding noise and prove that Kyber decoding noise can be bounded by a sphere. This result shows that the Kyber encoding problem is essentially a sphere packing in a hypercube. The original Kyber encoder uses the integer lattice for sphere packing purposes, which is far from optimal. Our second main contribution is to construct optimal lattice codes to ensure denser packing and a lower decryption failure rate (DFR).  Given the same ciphertext size as the original Kyber, the proposed lattice encoder enjoys a larger decoding radius, and is able to encode much more information bits. This way we achieve a decrease of
the communication cost by up to $32.6\%$, and a reduction of the DFR by a factor of up to $2^{85}$. Given the same plaintext size as the original Kyber, e.g., $256$ bits, we propose a bit-interleaved coded modulation (BICM) approach, which combines a BCH code and the proposed lattice encoder. The proposed BICM scheme significantly reduces the DFR of Kyber, thus enabling further compression of 
the ciphertext. Compared with the original Kyber encoder, the communication cost is reduced by $24.49\%$, while the DFR is decreased by a factor of $2^{39}$. The proposed encoding scheme is a constant-time algorithm, thus resistant against the timing side-channel attacks.}


%%\pacs[JEL Classification]{D8, H51}

%%\pacs[MSC Classification]{35A01, 65L10, 65L12, 65L20, 65L70}

\maketitle

\section{Introduction}\label{sec1}

The currently deployed Internet key exchange protocols use RSA and Elliptic Curve-based public-key cryptography (ECC), which have been proven to be vulnerable to quantum computing attacks. In July 2022, the National Institute of Standards and Technology (NIST) announced that a new post-quantum algorithm, Kyber, will replace current RSA and ECC-based key-establishment algorithms \cite{NISTpqc2022}. The NIST draft standard for Kyber was released in August 2023 \cite{NISTpqcdraft2023}. However, the quantum resistance of Kyber is achieved at a high communication cost, which refers to the ciphertext expansion rate (CER), i.e., the ratio of the ciphertext size to the plaintext size. Kyber introduces a $24-49$ times more communication cost compared to ECC-based algorithms, which limits its practical application for resource-constrained devices.

Kyber uses the module learning with errors problem (M-LWE) as its underlying mathematical problem \cite{Kyber2018}. Different from RSA and ECC, LWE-based encryption schemes, e.g., Kyber \cite{Kyber2018}, Saber \cite{Sabar2018}, and FrodoKEM \cite{FrodoKEM2021},  have a small decryption failure rate (DFR), where the involved parties fail to derive a shared secret key. Meanwhile, a high DFR might allow an adversary to recover the secret key using a number of failed ciphertexts \cite{DFRAttack2019}. For Kyber,  an attacker can use Grover search to find the ciphertexts with a slightly higher chance of producing a failure \cite{Kyber2021}. For Saber, D’Anvers and Batsleer have shown that the quantum security can be reduced from $172$ bits to $145$ bits by using a multitarget decryption failure attack \cite{DFRattack2022}. Therefore, LWE-based encryption schemes commonly choose their parameters so that
their DFR is small enough to avoid decryption failure attacks. The downside of these settings is an increased CER and a higher computational complexity.

Given a sufficiently small DFR, coding technologies have the potential to significantly reduce the CER of LWE-based encryption schemes. Some recent research models the decryption decoding problem in LWE-based encryption schemes as the decoding problem on additive white Gaussian noise (AWGN) channels \cite{NewhopeECC2018}\cite{Kyberpolar2022}\cite{FrodoCong2022}. Coding theory is then applied to improve the import performance parameters of LWE-based approaches, such as security level, DFR, and CER. In \cite{NewhopeECC2018}, a concatenation of BCH and LDPC codes was introduced to NewHope, to reduce the CER by a factor of $12.8\%$.  In \cite{Kyberpolar2022}, a rate-$1/2$ polar code was applied to Kyber, to reduce the DFR. In \cite{FrodoCong2022}, Barnes–Wall lattice codes were used in FrodoKEM, to reduce the CER by a factor of up to $6.7\%$. However, the decoding noise in Kyber contains a large uniform component: Kyber uses a compression function to round the $12$-bit ciphertext coefficients to $4$-bit (KYBER512/768) or $5$-bit (KYBER1024) integers. When decompressed, a $7$ or $8$-bit (almost) uniform noise will be added to the ciphertext \cite{Kyber2021}. Coding approaches for the AWGN channel may not be directly applicable to Kyber.

Despite its advantages, coding technologies add an extra decoding step to the LWE-based encryption schemes, which might be vulnerable to side-channel attacks.  A side-channel attack aims to gather information through physical channels such timing, power consumption, and electromagnetic emissions. In \cite{ECCTimingAttack2019}, the authors proposed a timing side-channel attack for BCH code used in the Ring-LWE scheme LAC. The proposed attack can distinguish between valid ciphertexts and failing ciphertexts, by using the fact that the original BCH decoder decodes valid codewords faster than the codewords that contain errors. However, this attack can be thwarted using a constant-time decoder. For BCH codes, a constant-time decoder was proposed in \cite{constantBCH2020}. For lattice codes, constant-time decoders have been proposed for Barnes–Wall lattice in \cite{FrodoCong2022} and Leech lattice in \cite{constanttimeLeech2016}. By eliminating the non-constant time execution in the decoding process, BCH codes, Barnes–Wall lattice codes, and Leech lattice codes are resistant against the timing attacks.

From an implementation perspective, most research focuses on reducing the execution time of Kyber. In \cite{KyberGPU2021}, the authors proposed a software implementation of Kyber using
a NVIDIA QUADRO GV100 graphics card, to increase the throughput for key exchange by $51$ times. In \cite{KyberHard2020}, the authors presented a pure hardware implementation of Kyber using Xilinx FPGAs, which achieves a maximum speedup of $129$ times for key exchange. In \cite{KyberSmartMeter2022}, the authtors proposed a hardware/software co-design implementation of Kyber in hardware-constrained smart meters. However, these implementations did not address the communication energy cost of Kyber. Since most Internet of Things (IoT) devices are battery-powered wireless sensors, energy efficiency is a crucial requirement for these devices. Kyber introduces a $24-49$ times more communication energy cost compared to current ECC-based algorithms. If Kyber is adopted for the
IoT devices, their battery life may be substantially reduced. There is clearly an urgent need to reduce the communication energy cost (i.e., CER) of Kyber.

In this work, we will demonstrate how to reduce the CER and DFR of Kyber. The main contribution of this paper is twofold: first, we show that the Kyber encoding problem is equivalent to the sphere packing in a hypercube. An explicit tail bound on the decoding noise magnitude is derived. This result allows us to upper bound the Kyber decoding noise as a hypersphere. Second, we propose a lattice-based encoder for Kyber, to ensure the densest packing of the noise sphere. Note that the original Kyber encoder uses a scaled integer lattice, which is far from optimal. Our construction is based on Barnes–Wall lattice and Leech lattice since they are denser than the integer lattice and have constant-time optimal decoders. For a fixed ciphertext size compared to Kyber, the proposed lattice encoder reduces the CER by up to $32.6\%$, and improves the DFR by a factor of up to $2^{85}$. For a fixed plaintext size compared to Kyber, e.g., $256$ bits, we propose a bit-interleaved coded modulation (BICM) approach, which combines a BCH code and a lattice encoder. Compared with the original Kyber encoder, the CER is reduced by $24.49\%$, while the DFR is decreased by a factor of $2^{39}$. The proposed encoding scheme is a constant-time algorithm, thus resistant against timing attacks.

\section{Preliminaries}
\subsection{Notation}
\emph{Rings:} Let $R$ and $R_{q}$ denote the rings $\mathbb{Z}[X]/(X^{n}+1)$ and $\mathbb{Z}_{q}[X]/(X^{n}+1)$, respectively. The degree $n$ of the monic polynomial is fixed to $256$ in Kyber.
Matrices and vectors are represented as bold upper-case and lower-case letters, respectively.
We use $\boldsymbol{\rm v}^T$ to represent the transpose of $\boldsymbol{\rm v}$.  We use the infinity norm $\|\boldsymbol{\rm v}\|_{\infty}$ to represent the magnitude of its largest entry.

\emph{Sampling and Distribution:} For a set $\mathcal{S}$, we write $s \leftarrow \mathcal{S}$ to denote that $s$
is chosen uniformly at random from $\mathcal{S}$. If $\mathcal{S}$ is a probability
distribution, then this denotes that $s$ is chosen according to
the distribution $\mathcal{S}$. For a polynomial $f(x) \in R_q$ or a vector of such polynomials, this notation is defined coefficient-wise. We use $\var(\mathcal{S})$ to represent the variance of the distribution $\mathcal{S}$. Let $x$ be a bit string and $S$ be a distribution taking $x$ as the input, then $y\sim S\coloneqq\mathsf{Sam}\left(x\right)$ represents that the output $y$ generated by distribution $S$ and input $x$ can be extended to any desired length. We denote $\beta_{\eta}=B(2\eta
,0.5)-\eta $ as the central binomial distribution over $\mathbb{Z}$. We denote $\mathcal{N}(\mu,\sigma^2)$ as the continuous normal distribution over $\mathbb{R}$, with mean $\mu$ and variance $\sigma^2$. We denote $\mathcal{U}(a, b)$ as the discrete uniform distribution over $\mathbb{Z}$, with minimum $a \in \mathbb{Z}$ and maximum $b \in \mathbb{Z}$.

\emph{Compress and Decompress:} Let $x\in\mathbb{R}$ be a real number, then $\left\lceil x\right\rfloor $ means rounding to the closet integer with ties rounded up. Let $x \in \mathbb{Z}_{q}$ and $d \in \mathbb{Z}$ be such that $2^d<q$. We define
\begin{align}
\mathsf{Compress}_{q}(x,d)&=\lceil (2^{d}/q)\cdot x\rfloor \bmod 2^{d},\nonumber\\
\mathsf{Decompress}_{q}(x,d)&=\lceil (q/2^{d})\cdot x\rfloor.\label{ComDecom}
\end{align}

\subsection{Kyber: Key Generation, Encryption, and Decryption}
Let $k, d_u, d_v$ be positive integer parameters, which are listed in Table \ref{Kyber_Par}. Let $\mathcal{M}_{2,n} = \{0,1\}^{n}$ denote the
message space, where every message $m \in \mathcal{M}_{2,n}$ can be
viewed as a polynomial in $R$ with coefficients in $\{0,1\}$.
Consider the public-key encryption scheme Kyber.CPA =
(KeyGen; Enc; Dec) as described in Algorithms 1 to 3 \cite{Kyber2021}. 
%The value of parameters are given in Table \ref{Kyber_Par}.
\vspace{-3mm}
\begin{algorithm}[H]
\caption{$\mathsf{Kyber.CPA.KeyGen()}$: key generation}
\label{alg:kyber_keygen}
\begin{algorithmic}[1]

    \State
    $\rho,\sigma\leftarrow\left\{ 0,1\right\} ^{256}$

    \State
    $\bA\sim R_{q}^{k\times k}\coloneqq\mathsf{Sam}(\rho)$

    \State
    $(\bs,\be)\sim\beta_{\eta_1}^{k}\times\beta_{\eta_1}^{k}\coloneqq\mathsf{Sam}(\sigma)$

    \State
    $\bt\coloneqq\boldsymbol{\mathrm{As+e}}$\label{line:t}

    \State \Return $\left(pk\coloneqq(\boldsymbol{\mathrm{t}},\rho),sk\coloneqq\bs\right)$  

\end{algorithmic}
\end{algorithm}

\vspace{-4mm}

\begin{algorithm}[H]
\caption{$\mathsf{Kyber.CPA.Enc}$ $(pk=(\boldsymbol{\mathrm{t}},\rho),m\in\mathcal{M}_{2,n})$}
\label{alg:kyber_enc}
\begin{algorithmic}[1]

	\State
	$r \leftarrow \{0,1\}^{256}$

	\State
	$\boldsymbol{\mathrm{A}}\sim R_{q}^{k\times k}\coloneqq\mathsf{Sam}(\rho)$
	
	\State  $(\boldsymbol{\mathrm{r}},\boldsymbol{\mathrm{e}_{1}},e_{2})\sim\beta_{\eta_1}^{k}\times\beta_{\eta_2}^{k}\times\beta_{\eta_2}\coloneqq\mathsf{Sam}(r)$
	
	\State  $\boldsymbol{\mathrm{u}}\coloneqq\mathsf{Compress}_{q}(\buone,d_{u})$\label{line:u}
	
	\State  $v\coloneqq\mathsf{Compress}_{q}(\boldsymbol{\mathrm{t}}^{T}\boldsymbol{\mathrm{r}}+e_2+\left\lceil {q}/{2}\right\rfloor \cdot m,d_{v})$\label{line:v}
	
	\State \Return $c\coloneqq(\boldsymbol{\mathrm{u}},v)$

\end{algorithmic}
\end{algorithm}

\vspace{-4mm}

\begin{algorithm}[H]
\caption{${\mathsf{Kyber.CPA.Dec}}\ensuremath{(sk=\bs,c=(\bu,v))}$}
\begin{algorithmic}[1]

    \State
    $\bu\coloneqq\mathsf{Decompress}_{q}(\bu,d_{u})$

    \State
    $v\coloneqq\mathsf{Decompress}_{q}(v,d_{v})$

    \State \Return $\mathsf{Compress}_{q}(v-\bs^{T}\bu,1)$

\end{algorithmic}
\end{algorithm}

\vspace{-4mm}
\begin{table}[ht]
\caption{Kyber Parameters in \cite{NISTpqcdraft2023}\cite{Kyber2021}: plaintext size $=256$ bits}
\label{Kyber_Par}\centering
\begin{tabular}{|c|c|c|c|c|c|c|c|c|c|c|}
\hline
& $n$ & $k$ & $q$ & $\eta_{1}$ & $\eta_{2}$ & $d_{u}$ & $d_{v}$ & $\delta$ & CER\\ \hline
KYBER512 & $256$ & $2$ & $3329$ & $3$ & $2$ & $10$ & $4$ & $2^{-139}$ & $24$\\ \hline
KYBER768 & $256$ & $3$ & $3329$ & $2$ & $2$ & $10$ & $4$ & $2^{-164}$ &$34$\\ \hline
KYBER1024 & $256$ & $4$ & $3329$ & $2$ & $2$ & $11$ & $5$ & $2^{-174}$ &$49$ \\ \hline
\end{tabular}
\end{table}

\subsection{Kyber Decoding Noise} 
Let $n_{e}$ be the decoding noise in
Kyber. According to \cite{Kyber2021}\cite{Kyberpolar2022}, we can write $n_{e}$ as%
\begin{align}
n_{e}&=v-\mathbf{s}^{T}\mathbf{u} -\left\lceil q/2\right\rfloor
\cdot m  \notag \\ 
&= \mathbf{e}^{T}\mathbf{r}+e_{2}+c_{v}-\mathbf{s}^{T}\left( \mathbf{e}%
_{1}+\mathbf{c}_{u}\right),  \label{Ne}
\end{align}
where $c_{v}\leftarrow \psi _{d_{v}}$, $\mathbf{c}_{u}\leftarrow \psi
_{d_{u}}^{k}$ are rounding noises generated due to the compression operation.
\begin{comment}
The distribution of $\psi _{d}\in R$ is defined as \cite{Kyber2018}.
\begin{enumerate}
\item Choose uniformly-random ${y}\leftarrow R$
\item Compute $y_D = \text{Decompress}_{q}\left( \text{Compress}%
_{q}\left( \mathbf{y},d\right) ,d\right)$
\item Return $\left( y-y_D \right) \mod ^{\pm}$ $q$
\end{enumerate}
\end{comment}
The elements in $c_v$ or $\mathbf{c}_u$ are assumed to be i.i.d. and independent of other terms in (\ref{Ne}). For a small $d$, e.g., $d=d_v$, we can assume that $\psi _{d}$ follows a discrete uniform distribution \cite{Kyber2018}:
\begin{equation}
c_{v}\leftarrow \psi _{d_v}\approx \mathcal{U}(-\lceil q/2^{d_v+1}\rfloor ,\lceil q/2^{d_v+1}\rfloor). \label{U_assumption}
\end{equation}

Kyber decoding problem can thus be formulated as%
\begin{equation}
y=v-\mathbf{s}^{T}\mathbf{u}=\left\lceil q/2\right\rfloor \cdot m+n_{e} \text{,}  \label{decoding_mode}
\end{equation}
i.e., given the observation $y\in R_{q}$, recover the value of $m$.

From a coding perspective, an interesting question is if we can model (\ref{decoding_mode}) as an AWGN channel. Given the values of $q$ and $d_v$ in Table \ref{Kyber_Par}, we observe that $n_e$ contains a uniform component $c_v$ with large variance:  $3640$ for KYBER512/768, and $918.67$ for KYBER1024. Therefore,  $n_{e}$ may not follow a normal distribution and hence \eqref{decoding_mode} may not be treated as an AWGN channel.

In Figure \ref{K_ND_Plot}, we compare the distribution of first coefficient in $n_{e}$, denoted as $n_{e,1}$, with the normal distribution using MATLAB function
$\text{normplot}()$. If the sample data has a normal distribution, then the data points appear along the reference line. With $100,000$ samples, we observe that the distribution of $n_{e,1}$ is different from the normal distribution. This result confirms that the AWGN assumption (\ref{decoding_mode}) of \cite{Kyberpolar2022} is not valid.

\begin{figure}[tbp]
\centering
\includegraphics[width=0.8\textwidth]{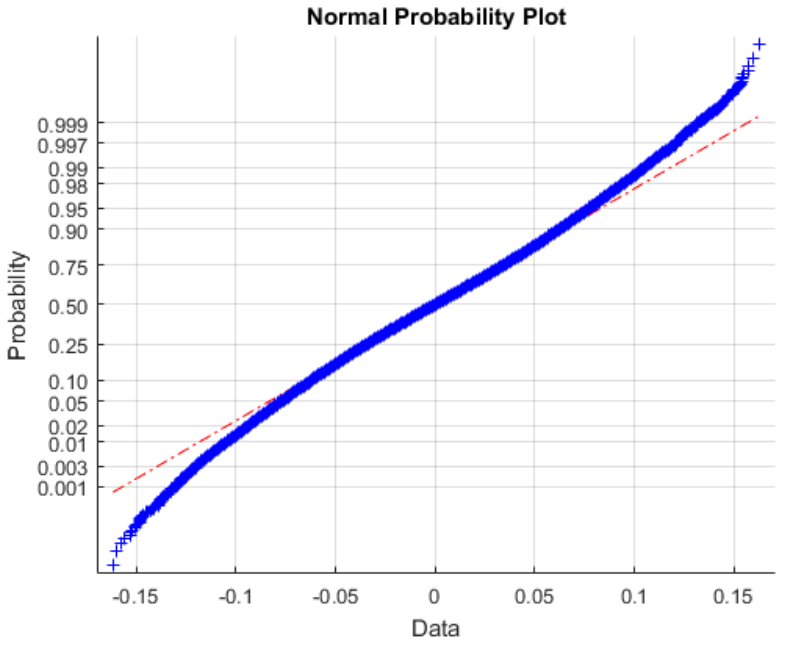} \vspace{0mm} 
\caption{KYBER{\protect\small 768: Comparing the distribution of $n_{e,1}$
to the normal distribution with $100,000$ samples.}}
\label{K_ND_Plot}
\end{figure}

\subsection{Decryption Failure Rate and Ciphertext Expansion Rate}
We let $\text{DFR} \triangleq \Pr (\hat{m} \neq m)$, where $\hat{m}$ is the decoded message. Since the original Kyber encodes and decodes $m$ bit-by-bit, DFR can be calculated by \cite{Kyber2018}:
\begin{equation}
\text{DFR} =  \delta :=\Pr (\|n_e\|_{\infty} \geq \lceil q/4 \rfloor). \label{KyberDFR}
\end{equation}
Note that it is desirable to have a small $\delta$, e.g., $\delta < 2^{-128}$,  in order to be safe against decryption failure attacks \cite{DFRAttack2019}. 

In this work, the communication cost refers to the ciphertext expansion rate (CER),
\begin{equation}
\text{CER} =  \frac{\# \text{ of bits in } c}{\# \text{ of bits in } m}, \label{CER}
\end{equation}
i.e., the ratio of the ciphertext size to the plaintext size. The values of $\delta$ and CER are given in Table \ref{Kyber_Par}.
For comparison purposes, we define the CER reduction ratio as
\begin{equation}
\text{CER-R}\triangleq1-\frac{\text{CER}_{\text{reduced}}}{\text{CER}}, \label{CER_R}
\end{equation}
where $\text{CER}$ represents the original CER of Kyber in Table \ref{Kyber_Par}, and  $\text{CER}_{\text{reduced}}$ represents the the reduced CER from this work.

\section{The Analysis on Kyber Decoding Noise } 
In this section, we study the distribution of $n_{e}$ in (\ref{Ne}). We will show that $n_{e}$ can be bounded by a hypersphere.
\vspace{-1mm}
\subsection{The Distribution of \texorpdfstring{$n_{e}$}{TEXT}}
The following lemma considers the product of two central binomial distributed polynomials. It is used to prove Theorem 1, the main result of this section.
\begin{lemma}
Suppose that $Z\leftarrow \beta_{\eta }$ and $Z^{\prime }\leftarrow \beta_{\eta
^{\prime }}$ are independent, then the polynomial product $ZZ^{\prime }$ (modulo $X^{n}+1
$) asymptotically approaches a multivariate normal distribution for large $n$, i.e.,%
\begin{equation*}
ZZ^{\prime }\leftarrow \mathcal{N}(0,n/4\cdot \eta\eta ^{\prime}I_{n}).
\end{equation*}
\end{lemma}

\begin{proof}
The proof is similar to that given in \cite[ Theorem 3]{CLTRLWE2022}. To be self-contained, we provide a detailed proof. Let $[Z_0, \ldots, Z_{n-1}]^T$ be the coefficients in $Z$, and $[Z^{\prime }_0, \ldots, Z^{\prime }_{n-1}]^T$ be the coefficients in $Z^{\prime }$.
Let $Y=ZZ^{\prime }$, we have that $Y$ has components $Y_{i}={\textstyle\sum\nolimits}%
_{j=0}^{n-1}\xi (i-j)Z_{i-j}Z_{j}^{\prime }$, where $\xi (x)=$ sign$%
(x)$ for $x\neq 0$ and $\xi (0)=1$ \cite[p.~338]{KyberNoise2022}. Note that $Z_{i-j}=Z_{(i-j)\bmod n}$%
. The mean of such a component $Y_{i}$ is given by%
\begin{equation*}
E(Y_{i})={\textstyle\sum\nolimits}_{j=0}^{n-1}E\left( \xi (i-j)Z_{i-j}Z_{j}^{\prime
}\right) =0 .
\end{equation*}%
The summands of a component $Y_{i}$ are independent, so the variance is
given by%
\begin{equation}
Var(Y_{i}) ={\textstyle\sum\nolimits}_{j=0}^{n-1}Var\left( \xi
(i-j)Z_{i-j}Z_{j}^{\prime }\right) =n/4\cdot \eta \eta ^{\prime } .\label{autocov}
\end{equation}

A similar argument shows that covariance of distinct components $Y_{i}$ and $%
Y_{i^{\prime }}$ for $i\neq i^{\prime }$ is given by
\begin{eqnarray}
Cov(Y_{i},Y_{i^{\prime }})
&=&E\left( {\textstyle\sum\nolimits}_{j=0}^{n-1}\xi (i-j)Z_{i-j}Z_{j}^{\prime }\cdot
{\textstyle\sum\nolimits}_{j=0}^{n-1}\xi (i^{\prime }-j)Z_{i^{\prime }-j}Z_{j}^{\prime
}\right)   \notag \\
&&-{\textstyle\sum\nolimits}_{j=0}^{n-1}E\left( \xi (i-j)Z_{i-j}Z_{j}^{\prime }\right)
{\textstyle\sum\nolimits}_{j=0}^{n-1}E\left( \xi (i^{\prime }-j)Z_{i^{\prime
}-j}Z_{j}^{\prime }\right)   \notag \\
&\overset{a}{=}&(\eta ^{\prime }/2)\cdot {\textstyle\sum\nolimits}_{j=0}^{n-1}\xi (i-j)\xi
(i^{\prime }-j)E\left( Z_{i-j}\right) E\left( Z_{i^{\prime }-j}\right)
\notag \\
&&+(\eta /2)\cdot{\textstyle\sum\nolimits}_{t=0}^{n-1}\xi (i-t)\xi (i^{\prime
}-t)E\left( Z_{i-t}^{\prime }\right) E\left(
Z_{i^{\prime }-t}^{\prime }\right)  \notag \\
&=&0. \label{cov}
\end{eqnarray}
The equality $(a)$ holds since for $(j\neq j^{\prime },i-j\neq i^{\prime }-j^{\prime
})$, 
\begin{eqnarray*}
E\left( \xi (i-j)Z_{i-j}Z_{j}^{\prime }\cdot
\xi (i^{\prime }-j^{\prime })Z_{i^{\prime }-j^{\prime }}Z_{j^{\prime
}}^{\prime }\right) 
&=&E\left( \xi (i-j)Z_{i-j}Z_{j}^{\prime }\right) \cdot
E\left( \xi (i^{\prime }-j^{\prime })Z_{i^{\prime }-j^{\prime }}Z_{j^{\prime
}}^{\prime }\right) , 
\end{eqnarray*}
i.e., $Z_{i-j}$, $Z_{j}^{\prime }$, $Z_{i^{\prime }-j^{\prime }}$, $%
Z_{j^{\prime }}^{\prime }$ are mutually independent. Therefore, we only need
to consider the components with $(j=j^{\prime },i-j\neq i^{\prime
}-j^{\prime })$ and $(j\neq j^{\prime },i-j=i^{\prime }-j^{\prime } \bmod n =t)$. We
have%
\begin{equation*}
Cov(Y_{i},Y_{i^{\prime }})=Cov(Y_{i},Y_{i^{\prime }})|_{j=j^{\prime
},i-j\neq i^{\prime }-j^{\prime }}+Cov(Y_{i},Y_{i^{\prime }})|_{j\neq
j^{\prime },i-j=i^{\prime }-j^{\prime } \bmod n =t}
\end{equation*}%
where%
\begin{eqnarray*}
&&Cov(Y_{i},Y_{i^{\prime }})|_{j=j^{\prime },i-j\neq i^{\prime }-j^{\prime }}\notag \\
&=&E\left( {\textstyle\sum\nolimits}_{j=0}^{n-1}\xi (i-j)\xi (i^{\prime
}-j)Z_{i-j}Z_{i^{\prime }-j}Z_{j}^{\prime }Z_{j}^{\prime }\right)  \\
&&-{\textstyle\sum\nolimits}_{j=0}^{n-1}E\left( \xi (i^{\prime }-j)Z_{i^{\prime
}-j}Z_{j}^{\prime }\right) E\left( \xi (i^{\prime }-j)Z_{i^{\prime
}-j}Z_{j}^{\prime }\right)  \\
&=&{\textstyle\sum\nolimits}_{j=0}^{n-1}\left( E\left( Z_{j}^{\prime 2}\right) -E\left(
Z_{j}^{\prime }\right) ^{2}\right) \xi (i-j)\xi (i^{\prime }-j)E\left(
Z_{i-j}\right) E\left( Z_{i^{\prime }-j}\right)  \\
&=&\eta ^{\prime }/2{\textstyle\sum\nolimits}_{j=0}^{n-1}\xi (i-j)\xi (i^{\prime
}-j)E\left( Z_{i-j}\right) E\left( Z_{i^{\prime }-j}\right),
\end{eqnarray*}
\begin{eqnarray*}
&&Cov(Y_{i},Y_{i^{\prime }})|_{j\neq j^{\prime },i-j=i^{\prime }-j^{\prime
} \bmod n = t}\notag \\
&=&E\left( {\textstyle\sum\nolimits}_{t=0}^{n-1}\xi (i-t)\xi (i^{\prime
}-t)Z_{t}Z_{t}Z_{i-t}^{\prime }Z_{i^{\prime }-t}^{\prime }\right) \notag \\
&&-{\textstyle\sum\nolimits}_{t=0}^{n-1}E\left( \xi (i-t)Z_{t}Z_{i-t}^{\prime }\right)
E\left( \xi (i^{\prime }-t)Z_{t}Z_{i^{\prime }-t}^{\prime }\right)  \\
&=&{\textstyle\sum\nolimits}_{t=0}^{n-1}\left( E\left( Z_{t}^{2}\right) -E\left(
Z_{t}\right) ^{2}\right) \xi (i-t)\xi (i^{\prime
}-t) E\left( Z_{i-t}^{\prime }\right) E\left(
Z_{i^{\prime }-t}^{\prime }\right)  \\
&=&\eta /2{\textstyle\sum\nolimits}_{t=0}^{n-1}\xi (i-t)\xi (i^{\prime
}-t)E\left( Z_{i-t}^{\prime }\right) E\left(
Z_{i^{\prime }-t}^{\prime }\right).
\end{eqnarray*}
According to (\ref{autocov}) and (\ref{cov}), $Y=ZZ^{\prime }$ has
covariance matrix $n/4\cdot \eta \eta ^{\prime }I_{n}$.

Furthermore, the summands of $Y_{i}$ are independent and identically
distributed, so a Central Limit argument shows that the distribution of $%
Y_{i}$ is well-approximated by a normal distribution for large $n$. Thus $Y$
can be approximated by a multivariate normal distribution. So we have $%
Y\leftarrow \mathcal{N}(0,n/4\cdot \eta \eta ^{\prime }I_{n})$.
\end{proof}

\begin{remark}
Lemma 1 shows the products $\mathbf{e}^{T}\mathbf{r}$ and $\mathbf{s}^{T} \mathbf{e}_{1}$ in (\ref{Ne}) asymptotically approach multivariate normal distributions for large $n$. This result can be easily generalised to the polynomial product of normal and binomial distributions, and two normal distributions as shown in \cite{CLTRLWE2022}.
\end{remark}

We then present Theorem 1, and show how it enables us to analyze the DFR and redesign the encoder and decoder for Kyber in the rest of the paper.
\begin{theo}
According to the Central Limit Theorem (CLT), the distribution of $n_e$ in (\ref{Ne}) asymptotically approaches the sum of multivariate normal and discrete uniform random variables for large $n$, i.e.,
\begin{equation}
n_{e}\leftarrow \mathcal{N}(0,\sigma _{G}^{2}I_{n}) +\mathcal{U}%
(-\lceil q/2^{d_{v}+1}\rfloor ,\lceil
q/2^{d_{v}+1}\rfloor ),  \label{D_ne}
\end{equation}%
where $\sigma _{G}^{2}=kn\eta _{{1}}^2/4+ kn\eta_{_1}/2 \cdot (\eta_2/2+\var(\psi_{d_u}))+\eta_2/2$. The values of $\var(\psi_{d_u})$ are listed in Table \ref{Kyber_Var}.
\end{theo}
\begin{proof}
According to Lemma 1,
with $n=256$, the distribution of $\mathbf{e}^{T}\mathbf{r}$ in (\ref{Ne}) is
well-approximated as a multivariate normal distribution, i.e., $\mathbf{e}^{T}\mathbf{r}\leftarrow \mathcal{N}(0,kn\eta _{{1}}^2/4\cdot I_{n})$.
Since $e_{2}\leftarrow \beta_{\eta_{2}}$ behaves very similarly to a (discrete) Gaussian and its variance $\eta _{2}/2$ is much smaller than $kn\eta _{{1}}^2/4$, the
distribution of $\mathbf{e}^{T}\mathbf{r}+e_{2}$ is well-approximated by%
\begin{equation}
\mathbf{e}^{T}\mathbf{r}+e_{2}\leftarrow \mathcal{N}(0,( kn\eta _{{1}}^2/4+\eta _{{2}}/2)I_{n}) . \label{D_ere}
\end{equation}%
We then estimate the distribution of $\mathbf{s}^{T}\left( \mathbf{e}_{1}+%
\mathbf{c}_{u}\right) $. According to \cite{Kyber2021}, the distribution of $\mathbf{e}_{1}+%
\mathbf{c}_{u}$ behaves very similarly to a Gaussian with the same variance $\eta_2/2+\var(\psi_{d_u}) $, i.e.,
\begin{equation}
\mathbf{e}_{1}+\mathbf{c}_{u}  \leftarrow
\mathcal{N}(0,(\eta_2/2+\var(\psi_{d_u}))I_{n}). \label{App_spec3}
\end{equation}
Using Remark 1, the distribution of $\mathbf{s}%
^{T}\left( \mathbf{e}_{1}+\mathbf{c}_{u}\right) $ is well-approximated as%
\begin{equation}
\setlength{\abovedisplayskip}{3pt} \setlength{\belowdisplayskip}{3pt}
\mathbf{s}^{T}\left( \mathbf{e}_{1}+\mathbf{c}_{u}\right) \leftarrow
\mathcal{N}(0,( kn\eta_{1}/2 \cdot (\eta_{2}/2+\var(\psi_{d_u}))) I_{n})  .\label{D_sec}
\end{equation}%
Given (\ref{D_ere}), (\ref{D_sec}), and (\ref{U_assumption}), for a large $n$, we can obtain (\ref{D_ne}).
\end{proof}

\begin{remark}
The idea of Theorem 1 and Lemma 1 is to approximate a discrete normal distribution by a continuous normal distribution. This approximation is commonly used in LWE literature \cite{FrodoCong2022}\cite{Kyber2021}. Theorem 1 shows that the coefficients in $n_e$ are independent identically distributed random variables.
\end{remark}

In Table \ref{Kyber_Var}, we show the CLT-based noise analysis for Kyber. We consider the normalised noise
term $n_{e}/\left\lceil q/2\right\rfloor$ and evaluate its variance. Note that the performance of the best known attacks against LWE-based encryption does not depend on the exact distribution of $n_e$, but rather on its variance. The values of $\var(\psi_{d_u})$ are computed numerically as explained in \cite{Kyber2021}. We observe that the estimated normalised variances match the simulation results.
 
\begin{table}[ht]
\caption{KYBER Decoding Noise Variance (Normalized): $10,000$ samples}
\label{Kyber_Var}\centering
\begin{tabular}{|c|c|c|c|c|}
\hline
& $\var(\psi_{d_u})$ & CLT  & Simulation \\ \hline
KYBER512 &$0.9$ \cite{Kyber2021}& $0.0023$ & $0.0023$ \\ \hline
KYBER768 &$0.9$ \cite{Kyber2021}& $0.0021$ & $0.0021$ \cite{Kyberpolar2022} \\ \hline
KYBER1024 & $0.38$ & $0.0012$ & $0.0012$ \\ \hline
\end{tabular}
\end{table}

In Figure \ref{K_UC_Plot}, we compare the distribution of the term $\hat{n}_{e}=\mathbf{e}^{T}\mathbf{r}+e_{2}-\mathbf{s}^{T} \mathbf{e}%
_{1}$ in (\ref{Ne}) with the normal distribution using MATLAB function
$\text{normplot}()$. This term can be viewed as the \emph{uncompressed} Kyber decoding noise, as $c_{v}= 0$ and $\mathbf{c}_{u}=\mathbf{0}$. Theorem 1 shows that for large $n$,
\begin{equation}
\hat{n}_{e}\leftarrow \mathcal{N}(0,(kn\eta _{{1}}^2/4+ kn\eta_{_1}\eta{_2}/4 +\eta_2/2)I_{n}).
\end{equation}
With $100,000$ samples, we observe that the distribution of the first element in $\hat{n}_{e}$, denoted as $\hat{n}_{e,1}$, follows a normal distribution. This result confirms that Theorem 1 is accurate for Kyber with $n=256$.

\begin{figure}[tbp]
\centering
\includegraphics[width=0.8\textwidth]{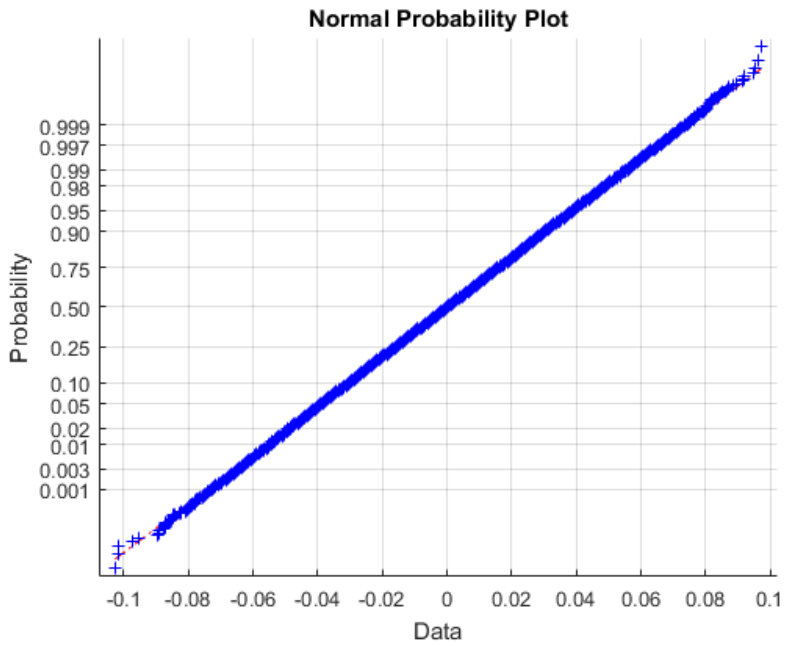} \vspace{0mm} 
\caption{Uncompressed KYBER{\protect\small 768: Comparing the distribution of $\hat{n}_{e,1}$
to the normal distribution with $100,000$ samples.}}
\label{K_UC_Plot}
\end{figure}

\subsection{Tail Bound, DFR, and Sphere Packing} 
\vspace{-1mm}
To gain more insight, we then estimate the magnitude of $n_{e}$ based on (\ref{D_ne}). Without loss of generality, we let
\begin{equation*}
n_{e}^{(\ell)}:=[n_{e,1},n_{e,2},\ldots,n_{e,\ell}]^{T}
\end{equation*}
to be the $\ell$ coefficients in $n_e$, where $1\leq \ell\leq n$. 
%A tail bound on $\Vert n_{e}^{\ell}\Vert$ is given below.

\begin{lemma}
$\Pr \left( \Vert n_{e}^{(\ell)}\Vert \leq z\right) \geq 1 - Q_{\ell/2}\left(
\frac{\sqrt{\ell}\lceil q/2^{d_{v}+1}\rfloor }{\sigma _{G}},\frac{z}{%
\sigma _{G}}\right)$, for $z \geq 0$, where $%
Q_{M}\left(a,b\right) $ is the generalised Marcum Q-function.
\end{lemma}

\begin{proof}
Given Theorem 1, we can write
$n_{e,i}=x_{i}+y_{i}$, where $x_{i}\leftarrow \mathcal{N}( 0,\sigma
_{G}^{2}) $ and $y_{i}\leftarrow \mathcal{U}(-\lceil
q/2^{d_{v}+1}\rfloor ,\lceil q/2^{d_{v}+1}\rfloor )$.
Let $\mathbf{y}=[y_{1},y_{2},\ldots , y_{\ell}]^{T}$ and $N_{U}=2\lceil
q/2^{d_{v}+1}\rfloor +1$. We have%
\begin{eqnarray*}
&&\Pr \left( \Vert n_{e}^{(\ell)}\Vert \leq z\right)  \\
&=&\Pr \left( \sqrt{{\textstyle\sum\nolimits}_{i=1}^{\ell}\left( x_{i}+y_{i}\right) ^{2}}%
\leq z\right)  \\
&=&{\textstyle\sum\nolimits}_{j=1}^{N_{U}^{\ell}}\Pr \left( \sqrt{{\textstyle\sum\nolimits}_{i=1}^{\ell}%
\left( x_{i}+\mu _{j,i}\right) ^{2}}\leq z\left\vert \mathbf{y}=\mathbf{\mu }%
_{j}\right. \right) \Pr \left( \mathbf{y}=\mathbf{\mu }_{j}\right)  \\
&=&\frac{1}{N_{U}^{\ell}}{\textstyle\sum\nolimits}_{j=1}^{N_{U}^{\ell}}\Pr \left( \sqrt{%
{\textstyle\sum\nolimits}_{i=1}^{\ell}\left( x_{i}+\mu _{j,i}\right) ^{2}}\leq z\left\vert
\mathbf{y}=\mathbf{\mu }_{j}\right. \right),
\end{eqnarray*}
where $\mathbf{\mu }_{j}=[\mu _{j,1},\mu _{j,2},\ldots ,\mu _{j,\ell}]^{T}$ is a
sample point in the sample space of $\mathbf{y}$.
Given $y_{i}=\mu _{j,i}$, $\sqrt{{\textstyle\sum\nolimits}_{i=1}^{\ell}\left( x_{i}+\mu
_{j,i}\right) ^{2}}$ follows non-central chi distribution, i.e.,%
\begin{align*}
\Pr \left( \sqrt{{{\textstyle\sum\nolimits}_{i=1}^{\ell}\left( x_{i}+\mu_{j,i}\right) ^{2}%
}/{\sigma _{G}^{2}}}\leq {z}/{\sigma _{G}}\left\vert \mathbf{y}_{j}=%
\mathbf{\mu }_{j}\right. \right) =1-Q_{\ell/2}\left( \sqrt{{%
{\textstyle\sum\nolimits}_{i=1}^{\ell}\mu _{j,i}^{2}}/{\sigma _{G}^{2}}},{z}/{\sigma
_{G}}\right),
\end{align*}%
where $Q_{M}\left( a,b\right) $ is the generalized Marcum Q-function. Since $%
Q_{M}\left( a,b\right) $ is strictly increasing in $a$ for all $a\geqslant 0$%
, we have%
\begin{align*}
\Pr \left( \Vert n_{e}^{(l)}\Vert \leq z\right)  &= 1-\frac{1}{N_{U}^{\ell}}%
{\textstyle\sum\nolimits}_{j=1}^{N_{U}^{\ell}}
Q_{\ell/2}\left( \sqrt{{%
{\textstyle\sum\nolimits}_{i=1}^{\ell}\mu _{j,i}^{2}}/{\sigma _{G}^{2}}},{z}/{\sigma
_{G}}\right)  \\
&\geqslant 1-Q_{\ell/2}\left( {\sqrt{\ell}/\sigma_{G} 
\cdot \lceil
q/2^{d_{v}+1}\rfloor },{z}/{\sigma _{G}}\right).
\end{align*}%
\end{proof}

\begin{remark}
With $\ell=1$ and $z=\left\lfloor q/4\right\rceil$, we can easily show
\begin{equation*}
\Pr \left( \Vert n_{e,1}\Vert \leq \left\lfloor q/4\right\rceil \right) \geq 1-2Q\left( \left(\lfloor q/4\rceil -\lceil q/2^{d_{v}+1}\rfloor \right) /\sigma
_{G}\right),
\end{equation*}
where $Q(a)$ is the Q-function. It gives an explicit upper bound on the DFR of the original Kyber encoder in (\ref{KyberDFR}):
\begin{equation*}
\delta \leq 1-\left(1-2Q\left( \left(\lfloor q/4\rceil -\lceil q/2^{d_{v}+1}\rfloor \right) /\sigma
_{G}\right)\right)^n.
\end{equation*}
We have $\delta \leq 2^{-142}, 2^{-167}, 2^{-176}$, for KYBER512/768/1024, respectively. Note that the values of $\delta$ in Table \ref{Kyber_Par} are evaluated numerically in \cite{Kyber2018}\cite{Kyber2021}. No closed-form expression is provided.
\end{remark}

\begin{sloppypar}
Lemma 2 shows that
Kyber decoding noise $n_{e}^{(\ell)}$ can be bounded with probability $1-Q_{\ell/2}\left(
{\sqrt{\ell}\lceil q/2^{d_{v}+1}\rfloor }/{\sigma _{G}},{z}/{%
\sigma _{G}}\right)
$ by a hypersphere with radius $z$. Since the addition in (\ref{decoding_mode}) is over the modulo $q$ domain, we can view the Kyber encoding problem as
a sphere packing problem: an arrangement of non-overlapping spheres within
a hypercube $\mathbb{Z}_{q}^{\ell}$. The original Kyber uses the lattice codes $\lfloor q/2\rceil\mathbb{Z}_{2}^{\ell}$ for sphere packing purposes, which is far from optimal. Even for very small dimensions $\ell$,
there exist much denser lattices that maintain the same minimal distance
between points. In the next section, we will construct lattice codes with denser packing and lower DFR.
\end{sloppypar}

\section{Lattice Codes for Kyber}
\subsection{Lattice Codes, Hypercube Shaping, and CVP Decoding}
\begin{definition}[Lattice]
An $\ell$-dimensional lattice $\Lambda$ is a
discrete additive subgroup of $\mathbb{R}^m$, $m \leq \ell$. Based on $\ell$ linearly independent vectors $b_1, \ldots, b_\ell$ in $\mathbb{R}^m$, $\Lambda$ can be
written as
\begin{equation*}
\Lambda =\mathcal{L}(\mathbf{B})=z_{1}\mathbf{b}_{1}+\cdots z_{\ell}\mathbf{b}%
_{\ell},
\end{equation*}%
where $z_{1},\ldots,z_{\ell}\in \mathbb{Z}$, and $\mathbf{B}=[\mathbf{b}_{1},\ldots,%
\mathbf{b}_{\ell}]$ is referred to as a generator matrix of $\Lambda$.
\end{definition}

\begin{definition}[Lattice Code]
A lattice code $\mathcal{C}(\Lambda, \mathcal{P})$ is the finite set of points
in $\Lambda$ that lie within the region $\mathcal{P}$:
\begin{equation*}
\mathcal{C}(\Lambda, \mathcal{P}) = \Lambda \cap \mathcal{P}.
\end{equation*}
If $\mathcal{P}=\mathbb{Z}_p^\ell$, the code $\mathcal{C}(\Lambda, \mathbb{Z}_p^\ell)$ is said to be generated from hypercube shaping (HS).
\end{definition}

\begin{definition}[CVP Decoder]
Given an input $\mathbf{y} \in \mathbb{R}^\ell$, the Closest Vector Problem (CVP) decoder returns the closet lattice vector to $\mathbf{y}$ over the lattice $\mathcal{L}(\mathbf{B})$, i.e.,
\begin{equation*}
\mathbf{x}=\mathsf{CVP}(\mathbf{y}, \mathcal{L}(\mathbf{B}))=\arg\min_{\mathbf{x}' \in \mathcal{L}(\mathbf{B})} \| \mathbf{x}'-\mathbf{y} \|.
\end{equation*}
\end{definition}

\begin{definition}[HS Encoder \cite{FrodoCong2022}]
Considering the Smith Normal Form factorization (SNF) of a lattice basis $\mathbf{B}$, denoted as $\mathbf{B}= \mathbf{U}\cdot \diag (\pi_1, \ldots, \pi_\ell) \cdot \mathbf{U}'$, where $\mathbf{U}, \mathbf{U}' \in \mathbb{Z}^{\ell \times \ell}$ are unimodular matrices.
Let the message space be
\begin{equation}
\mathcal{M}_{p,\ell}=\left\{ 0,1,\ldots,p/\pi_1-1\right\} \times \cdots \times
\left\{ 0,1,\ldots,p/\pi_\ell-1\right\} , \label{m_p}
\end{equation}%
where $p > 0$ is a common multiplier of $\pi_1, \ldots, \pi_\ell$. Given $\mathbf{m} \in \mathcal{M}_{p, \ell}$, a HS encoder returns a codeword $\mathbf{x} \in \mathcal{C}(\mathcal{L}(\mathbf{B}), \mathbb{Z}_p^\ell)$:
\begin{equation}
\mathbf{x} = \hat{\mathbf{B}}\mathbf{m} \bmod p , \label{HC_shaping}
\end{equation}
where $\hat{\mathbf{B}}= \mathbf{U}\cdot \diag (\pi_1, \ldots, \pi_\ell)$.
\end{definition}

\begin{comment}
Let $\mathbf{B}= \mathbf{U}\cdot \diag (\pi_1, ..., \pi_l) \cdot \mathbf{U}'$ be the Smith Normal Form factorization (SNF) of a lattice basis $\mathbf{B}$,
where $\mathbf{U}, \mathbf{U}' \in \mathbb{Z}^{l \times l}$ are unimodular matrices. The \emph{rectangular form} basis of the lattice $\mathcal{L}(\mathbf{B})$ is defined by
\begin{equation*}
\hat{\mathbf{B}}= \mathbf{U}\cdot \diag (\pi_1, ..., \pi_l).
\end{equation*}
\end{comment}

\begin{definition}[HS-CVP Decoder \cite{FrodoCong2022}]
Let $\mathbf{y}$ be a noisy version of $\mathbf{x}\in \mathcal{C}(\mathcal{L}(\mathbf{B}), \mathbb{Z}_p^\ell) $. The HS-CVP decoder returns an estimated message $ \hat{\mathbf{m}} =[\hat{m}_1, \ldots, \hat{m}_\ell]^T \in \mathcal{M}_{p,l}$:
\begin{align}
\hat{\mathbf{m}} &=\mathsf{HS-CVP}(\mathbf{y}, \mathcal{L}(\mathbf{B}))  \notag \\ &=\hat{\mathbf{B}}^{-1}\cdot\mathsf{CVP}(\mathbf{y}, \mathcal{L}(\mathbf{B})) \bmod (p/\pi_1, \ldots, p/\pi_\ell), \label{HC_decoding}
\end{align}%
where $\hat{m}_i= (\hat{\mathbf{B}}^{-1}\mathsf{CVP}(\mathbf{y}, \mathcal{L}(\mathbf{B}))_i \bmod p/\pi_i$, for $i= 1,\ldots, \ell$.
\end{definition}

For the choice of $\mathcal{L}(\mathbf{B})$, in this work, we consider Barnes–Wall lattice with $\ell=16$ (BW16) and Leech lattice with $\ell=24$ (Leech24)\cite{BK:Conway93}. They provide the optimal packing in the corresponding dimensional space and have constant-time CVP decoders \cite{FrodoCong2022}\cite{constanttimeLeech2016}. Since the coefficients in Kyber are integers, we will scale the original generator matrix to an integer matrix $\mathbf{B} \in \mathbb{Z}^{\ell \times \ell}$, and use the corresponding $\hat{\mathbf{B}}$.

\begin{example} For BW16, the matrix $\mathbf{B}$ is given by \cite{FrodoCong2022}
\begin{equation}
\setlength{\arraycolsep}{3pt}
\mathbf{B} = \begin{bNiceMatrix}[r, columns-width=auto]
1 & 1 & 1 & 1 & 1 & 2 & 2 & 2 & 2 & 2 & 2 & 2 & 2 & 2 & 2 & 4 \\
1 & 1 & 1 & 1 & 0 & 2 & 2 & 0 & 2 & 0 & 0 & 2 & 0 & 0 & 0 & 0 \\
1 & 1 & 1 & 0 & 1 & 2 & 0 & 2 & 0 & 2 & 0 & 0 & 2 & 0 & 0 & 0 \\
1 & 1 & 1 & 0 & 0 & 2 & 0 & 0 & 0 & 0 & 0 & 0 & 0 & 0 & 0 & 0 \\
1 & 1 & 0 & 1 & 1 & 0 & 2 & 2 & 0 & 0 & 2 & 0 & 0 & 2 & 0 & 0 \\
1 & 1 & 0 & 1 & 0 & 0 & 2 & 0 & 0 & 0 & 0 & 0 & 0 & 0 & 0 & 0 \\
1 & 1 & 0 & 0 & 1 & 0 & 0 & 2 & 0 & 0 & 0 & 0 & 0 & 0 & 0 & 0 \\
1 & 1 & 0 & 0 & 0 & 0 & 0 & 0 & 0 & 0 & 0 & 0 & 0 & 0 & 0 & 0 \\
1 & 0 & 1 & 1 & 1 & 0 & 0 & 0 & 2 & 2 & 2 & 0 & 0 & 0 & 2 & 0 \\
1 & 0 & 1 & 1 & 0 & 0 & 0 & 0 & 2 & 0 & 0 & 0 & 0 & 0 & 0 & 0 \\
1 & 0 & 1 & 0 & 1 & 0 & 0 & 0 & 0 & 2 & 0 & 0 & 0 & 0 & 0 & 0 \\
1 & 0 & 1 & 0 & 0 & 0 & 0 & 0 & 0 & 0 & 0 & 0 & 0 & 0 & 0 & 0 \\
1 & 0 & 0 & 1 & 1 & 0 & 0 & 0 & 0 & 0 & 2 & 0 & 0 & 0 & 0 & 0 \\
1 & 0 & 0 & 1 & 0 & 0 & 0 & 0 & 0 & 0 & 0 & 0 & 0 & 0 & 0 & 0 \\
1 & 0 & 0 & 0 & 1 & 0 & 0 & 0 & 0 & 0 & 0 & 0 & 0 & 0 & 0 & 0 \\
1 & 0 & 0 & 0 & 0 & 0 & 0 & 0 & 0 & 0 & 0 & 0 & 0 & 0 & 0 & 0 %
\end{bNiceMatrix}. \label{BW16B}
\end{equation}
The corresponding matrix $\hat{\mathbf{B}}$ is given by
\begin{equation}
\setlength{\arraycolsep}{3pt}
\hat{\mathbf{B}}=\begin{bNiceMatrix}[r, columns-width=auto]
1 & 0 & 0 & 0 & 0 & 0 & 0 & 0 & 0 & 0 & 0 & 0 & 0 & 0 & 0 & 0 \\
1 &-1 & 0 & 0 & 0 & 0 & 0 & 0 & 0 & 0 & 0 & 0 & 0 & 0 & 0 & 0 \\
1 & 0 &-1 & 0 & 0 & 0 & 0 & 0 & 0 & 0 & 0 & 0 & 0 & 0 & 0 & 0 \\
1 &-1 &-1 & 0 & 0 & 0 & 0 & 2 & 2 & 0 & 0 & 0 & 0 & 0 & 0 & 0 \\
1 & 0 & 0 &-1 & 0 & 0 & 0 & 0 & 0 & 0 & 0 & 0 & 0 & 0 & 0 & 0 \\
1 &-1 & 0 &-1 & 0 & 2 & 0 & 0 & 0 & 0 & 0 & 0 & 0 & 0 & 0 & 0 \\
1 & 0 &-1 &-1 & 0 & 0 & 2 & 0 & 0 & 0 & 0 & 0 & 0 & 0 & 0 & 0 \\
1 &-1 &-1 &-1 & 0 & 2 & 2 & 2 & 0 & 0 & 0 & 0 & 0 & 0 & 0 & 0 \\
1 & 0 & 0 & 0 &-1 & 0 & 0 & 0 & 0 & 0 & 0 & 0 & 0 & 0 & 0 & 0 \\
1 &-1 & 0 & 0 &-1 & 0 & 0 & 0 & 0 & 2 & 0 & 0 & 0 & 0 & 0 & 0 \\
1 & 0 &-1 & 0 &-1 & 0 & 0 & 0 & 0 & 0 & 2 & 0 & 0 & 0 & 0 & 0 \\
1 &-1 &-1 & 0 &-1 & 0 & 0 & 2 & 2 & 2 & 2 &-2 & 0 & 0 & 0 & 0 \\
1 & 0 & 0 &-1 &-1 & 0 & 0 & 0 & 0 & 0 & 2 &-2 & 2 & 0 & 0 & 0 \\
1 &-1 & 0 &-1 &-1 & 2 & 0 & 0 & 0 & 2 & 2 &-2 & 0 & 2 & 0 & 0 \\
1 & 0 &-1 &-1 &-1 & 0 & 2 & 0 & 0 & 0 & 2 & 0 & 2 & 0 & 2 & 0 \\
1 &-1 &-1 &-1 &-1 & 2 & 2 & 2 & 0 & 2 & 2 &-2 & 0 & 2 & 2 & 4
\end{bNiceMatrix}, \label{BW16Bh}
\end{equation}
The values of $\pi_1 \ldots \pi_{\ell=16}$ are given by
\begin{equation}
\left[\pi_1 \ldots \pi_{16} \right]=\left[1 , 1 , 1 , 1 , 1 , 2 , 2 , 2 , 2 , 2 , 2 , 2 , 2 , 2 , 2 , 4 \right]. \label{BW16pi}
\end{equation}
\end{example}

\subsection{Lattice Encoder and Decoder for Kyber}
%The encryption and decryption processes are given below.
We divide the $n$ coefficients in $m$ to $\kappa =n/\ell$ blocks:
\begin{equation}
m=[m_{\ell,1},\ldots, m_{\ell,\kappa}]^T \in \mathcal{M}_{p, \ell}^\kappa, \label{m_block_ele}
\end{equation}
where $m_{\ell,i}=[m_{\ell,i,1},\ldots, m_{\ell,i,\ell}]^T \in \mathcal{M}_{p, \ell}$ in (\ref{m_p}) is the $i^{\text{th}}$ block of coefficients in $m$, $i=1,\ldots, \kappa$. From the viewpoint of lattice, the model in (\ref{decoding_mode}) can be generalized to
\begin{align}
y&= \left\lceil q/p\right\rfloor \cdot[\hat{\mathbf{B}}m_{\ell,1} \bmod p,\ldots, \hat{\mathbf{B}}m_{\ell,\kappa} \bmod p]^T +n_{e} , \label{decoding_mode_lattice}
\end{align}
where $\hat{\mathbf{B}}$ is given in Definition 4,  and $\hat{\mathbf{B}}m_{\ell,i} \bmod p  \in  \mathcal{C}(\mathcal{L}(\mathbf{B}), \mathbb{Z}_p^\ell)$, for $i=1,\ldots, \kappa$. In other words, the message $m \in \mathcal{M}_{p, \ell}^\kappa$ is encoded to $\kappa$ lattice points.

We describe the lattice encoder and decoder in Algorithms 4 and 5, where the encoder $\Lambda$-$\mathsf{Enc}$ takes as the input a message $m \in \mathcal{M}_{p, \ell}^\kappa$ and outputs $v$ (i.e., the second part of the ciphertext), and where the decoder $\Lambda$-$\mathsf{Dec}$ takes as input a ciphertext $(\bu,v)$ and outputs a message.

\vspace{-1mm}
\begin{algorithm}[H]
\caption{$\Lambda$-$\mathsf{Enc}(m \in \mathcal{M}_{p, \ell}^\kappa)$: lattice encoder \\//Replace Step 5 in \textbf{Algorithm 2}}
\label{lattice_enc}
\begin{algorithmic}[1]

    \State
    $x_{\ell,i}:=\hat{\mathbf{B}}\cdot m_{\ell,i} \bmod p$, $i=1,\ldots, \kappa$

    \State
    $x:=[x_{\ell,1},\ldots, x_{\ell,\kappa}]^T$

    \State
   \Return  $v:= \mathsf{Compress}_{q}( \mathbf{t}^{T}\mathbf{r}+e_{2}+\lfloor q/p\rceil \cdot x,d_{v}) $

\end{algorithmic}
\end{algorithm}

\vspace{-5mm}

\begin{algorithm}[H]
\caption{$\Lambda$-$\mathsf{Dec}(\bu,v)$: lattice decoder \\//Replace Step 3 in \textbf{Algorithm 3}}
\label{lattice_dec}
\begin{algorithmic}[1]

    \State
    $y:=v -\mathbf{s}^{T}\cdot\mathbf{u} = [y_{\ell,1},\ldots, y_{\ell,\kappa}]^T$, where $y_{\ell,i}$ is the $i^{\text{th}}$ block of coefficients in $y$

    \State
    ${m}_{\ell,i} :=\mathsf{HS-CVP}(y_{\ell,i}/\lfloor q/p\rceil, \mathcal{L}(\mathbf{B})),  i=1, \ldots, \kappa$

    \State
   \Return  ${m}:=[{m}_{\ell,1}, \ldots, {m}_{\ell,\kappa}]^T$

\end{algorithmic}
\end{algorithm}

\vspace{-3mm}

\begin{remark}
With $\mathcal{L}(\mathbf{B})= \mathbb{Z}^{\ell}$, $\hat{\mathbf{B}}=I_\ell$, and $p=2$, the proposed lattice encoder and decoder reduce to the original ones used in Kyber.
\end{remark}

\subsection{CER and DFR Reduction}
The number of information bits for each block is given by
\begin{equation}
b(\ell,p)={\textstyle\sum\nolimits}_{i=1}^{\ell}\log _{2}(p/\pi_{i}), \label{b_pl}
\end{equation}%
where $\pi_{i}$ is given in Definition 4, for $i=1,\ldots, \ell$. The total number of encrypted bits can be calculated by 
\begin{equation}
N=\kappa b(\ell,p). \label{num_ebits}
\end{equation}
As shown in Step 3 in Algorithm 4, the ciphertext size remains unchanged. The CER-R in (\ref{CER_R}) can be calculated by:
\begin{equation*}
\text{CER-R}=1-n/N=1-l/b(\ell,p). 
\end{equation*}

Let $\lambda (p)$ be the length of a shortest non-zero vector in the lattice $\mathcal{L}(\lfloor
q/p\rceil\mathbf{B})$. The correct decoding radius of HS-CVP decoder is $\lambda (p)/2$. Let $\delta_{i}$ be the decryption failure rate for the $i^{\text{th}}$ block of coefficients ${m}_{\ell,i}$ in (\ref{m_block_ele}). Using Lemma 2, we have
\begin{equation*}
\delta_i\leq  Q_{\ell/2}\left({\sqrt{\ell}/\sigma _{G} \cdot \lceil q/2^{d_{v}+1}\rfloor }, \lambda (p)/(2\sigma _{G})\right).
\end{equation*}%
According to Theorem 1, the noise elements in (\ref{decoding_mode_lattice}) are independent from block to block. Therefore, the decryption failure rate of $\kappa $ blocks can be calculated by 
\begin{equation}
\delta = 1-\prod_{i-1}^{\kappa}(1-\delta_i) \approx \kappa Q_{\ell/2}\left(
{\sqrt{\ell}/\sigma _{G} \cdot \lceil q/2^{d_{v}+1}\rfloor }, \lambda (p)/(2\sigma _{G})\right). \label{DFR_L}
\end{equation}%

Table \ref{KYBER_LC} compares the normalized correct decoding radius $\lambda(p)/(2\left\lfloor q/2\right\rceil)$\thinspace
, CER-R, $N$, and $\delta$ for difference lattices. We observe a tradeoff between CER and DFR from Table \ref{KYBER_LC}: a smaller CER requires a larger $p$, which will reduce the correct decoding radius of $\mathcal{L}(\lfloor
q/p\rceil\mathbf{B})$ accordingly. Interestingly, KYBER1024 enjoys a much lower DFR compared to KYBER512/768. This can be explained by the smallest noise variance in Table \ref{Kyber_Var}. 

\begin{table}[th]
\centering
\caption{Lattice Codes for KYBER}
\label{KYBER_LC}\centering
\begin{tabular}{|c|c|c|c|}
\hline
Lattice & $\mathbb{Z}^{\ell}$ \cite{Kyber2021}  & $\text{BW16}$ & $\text{Leech24}$ \\ \hline
$p$ & $2$ & $4$ & $8$ \\ \hline
$\lambda(p)/(2\left\lfloor q/2\right\rceil )$ & $0.5$ &$0.7067$ & $0.7067$ \\ \hline
$b(\ell,p)/\ell$ & $1$ & $20/16$ & $36/24$ \\ \hline
$N$ & $256$ &$320$ & $380\footnotemark[1]$ \\ \hline
CER-R & $0\%$ & $20\%$ & $32.6\%$ \\ \hline
$\delta$: KYBER512 & $2^{-139}$ & $2^{-149}$ & $2^{-111}$ \\ \hline
$\delta$: KYBER768 & $2^{-164}$ & $2^{-177}$ & $2^{-131}$ \\ \hline
$\delta$: KYBER1024 & $2^{-174}$ & $2^{-259}$ & $2^{-226}$ \\ \hline
\end{tabular}%
\footnotetext[1]{Since $n=256=10\times 24+16$, we consider $10$ Leech24 codewords with $p=8$ and $1$ BW16 codeword with $p=4$}
%\tablefootnote{asdfafds}
\end{table}

\subsection{Decoding Complexity, Security, and Observation}
As shown in (\ref{HC_decoding}), the decoding complexity is determined by the CVP decoder for the selected lattice. For \text{BW16} lattice, the running time of the constant-time CVP decoder is proportional to $32\ell$ \cite{FrodoCong2022}. For the Leech lattice, the CVP decoder requires $3,595$ real operations in the worst case \cite{VD93Leech}. The constant-time implementation of \cite{VD93Leech} adds $34.5\%$ of overhead to the algorithm \cite{constanttimeLeech2016}. The proposed lattice encoder is resistant to the timing side-channel attacks. Since lattice encoder/decoding doesn't affect the M-LWE hardness assumption in Kyber \cite{Kyber2021}, the security arguments remain the same.

\section{Reducing CER and DFR with a Fixed Plaintext Size}
In the previous section, we have shown that lattice codes can reduce CER by increasing the plaintext size. In many applications, the plaintext size is fixed, e.g., 256 bits. In this section, we will explain how to reduce the CER and DFR of Kyber, with a fixed plaintext size $K=256$.
\subsection{Bit-Interleaved Coded Modulation for Kyber}
We propose a bit-interleaved coded modulation (BICM) approach, which combines an error-correcting code and a lattice encoder. The concept of BICM was originally proposed for the Rayleigh fading channel \cite{BICM1998}. It combines coding with
modulation in a situation where the performance of an error-correcting code depends on its minimum Hamming distance, rather than on the minimum Euclidean distance of the modulation symbols. This is achieved by bit-wise interleaving at the error-correcting encoder output. Since the lattice encoder can be viewed as a form of modulation \cite{FrodoCong2022}, we can apply error correcting codes combined with lattice encoder, by using the extra $N-K$ bits (enabled by the proposed lattice encoder) as the parity check bits of an error correcting code. The values of $N$ are given in Table \ref{KYBER_LC}.

\begin{definition}[Bit Mapper and Demapper]
The BICM approach involves bit mapper and demapper, where the former maps binary bits to an integer vector, while the latter performs the
inverse operation. Let $b_{\ell,i} \in \mathcal{M}_{2, b(\ell,p)}$ be the bit representation of $m_{\ell,i}\in \mathcal{M}_{p, \ell}$ in (\ref{m_block_ele}), for $i=1,\ldots, \kappa$. The bit mapper $\mathcal{M}_{2, b(\ell,p)} \mapsto \mathcal{M}_{p, \ell}$ and demapper $\mathcal{M}_{p, \ell} \mapsto \mathcal{M}_{2, b(\ell,p)}$ are defined by
\begin{align}
\mathsf{bit2int}(b_{\ell,i})&=m_{\ell,i}=[m_{\ell,i,1},\ldots, m_{\ell,i,\ell}]^T  \nonumber \\
\mathsf{int2bit}(m_{\ell,i})&=b_{\ell,i}=[b_{\ell,i,1},\ldots, b_{\ell,i, b(\ell,p)}]^T, \label{b_l_i}
\end{align}
where the $j^{\text{th}}$ element in $m_{\ell,i}$, denoted as $m_{\ell,i,j}$, is converted to the $\log_{2}(p/\pi_{j})$ bits:
\begin{equation*}
\left[b_{\ell,i,{\textstyle\sum\nolimits}_{w=1}^{j-1}\log _{2}(p/\pi_{w})+1},\ldots,b_{\ell,i,{\textstyle\sum\nolimits}_{w=1}^{j}\log _{2}(p/\pi_{w})}\right]^T,
\end{equation*}
and vice versa, for $j=1, \ldots. \ell$. The value of  $b(\ell,p)$ is given in (\ref{b_pl}), and $\pi_{j}$ is given in Definition 4, for $j=1,\ldots, \ell$.
\end{definition}

\begin{definition}[BCH Code \cite{constantBCH2020}]
The BICM approach involves an error-correcting code. To avoid the timing side-channel attacks, in this work, we choose BCH code, which has a constant-time decoder \cite{constantBCH2020}. Let BCH $(N,K)$ be a BCH code, where $N=\kappa b(\ell,p)$ is the number of encrypted bits in (\ref{num_ebits}), and $K=256$ represents the number of information bits. The number of correctable errors denoted as $t$, can be determined during the construction of the code. Given a message $m \in \mathcal{M}_{2,K}$, the BCH encoder and decoder are defined by
\begin{align}
[b_{\ell,1}, \cdots, b_{\ell,\kappa}]^T &=\mathsf{E}\text{-}\mathsf{Enc}(m)\nonumber \nonumber \\
m&=\mathsf{E}\text{-}\mathsf{Dec}([b_{\ell,1}, \cdots, b_{\ell,\kappa}]^T), \label{ECC_ENC}
\end{align}
where $b_{\ell,i}$ is defined in (\ref{b_l_i}), for $i=1,\ldots, \kappa$.
\end{definition}

\begin{definition}[Bit Interleaver and Deinterleaver]
The BICM approach involves bit interleaver and deinterleaver, where the former randomly changes the order of a binary sequence, while the latter performs the  inverse operation. The bit interleaver and deinterleaver are defined by
\begin{align}
[\hat{b}_{\ell,1}, \cdots, \hat{b}_{\ell,\kappa}]^T&=\Omega([b_{\ell,1}, \cdots, b_{\ell,\kappa}]^T)\nonumber \\
[{b}_{\ell,1}, \cdots, {b}_{\ell,\kappa}]^T&= \Omega^{-1}([\hat{b}_{\ell,1}, \cdots, \hat{b}_{\ell,\kappa}]^T), \label{interleaver}
\end{align}
where $b_{\ell,i}$ is defined in (\ref{b_l_i}), and $\hat{b}_{\ell,i}$ represents the permuted bits, for $i=1,\ldots, \kappa$.

\end{definition}

\begin{definition}[BICM]
The BICM models are represented by the block diagram of Figure \ref{K_BICM}, which includes 1) an error-correcting code $\mathsf{E}$-$\mathsf{Enc}$ in (\ref{ECC_ENC}); 2) an interleaver $\Omega$ in (\ref{interleaver}); 3) a lattice encoder $\Lambda$-$\mathsf{Enc}$ in Algorithm 4,
modelled by a labelling map $ b_{\ell,i}$ in (\ref{b_l_i}) and a lattice point $x_{\ell,i} \in \mathcal{C}(\mathcal{L}(\mathbf{B}), \mathbb{Z}_p^\ell)$, $i=1,\ldots, \kappa$; 4) a lattice decoder $\Lambda$-$\mathsf{Dec}$ in Algorithm 5; 5) a deinterleaver $\Omega^{-1}$ in (\ref{interleaver}); 6) an error-correcting decoder $\mathsf{E}$-$\mathsf{Dec}$  in (\ref{ECC_ENC}).
\end{definition}

\begin{figure}[tbp]
\centering
\includegraphics[width=0.8\textwidth]{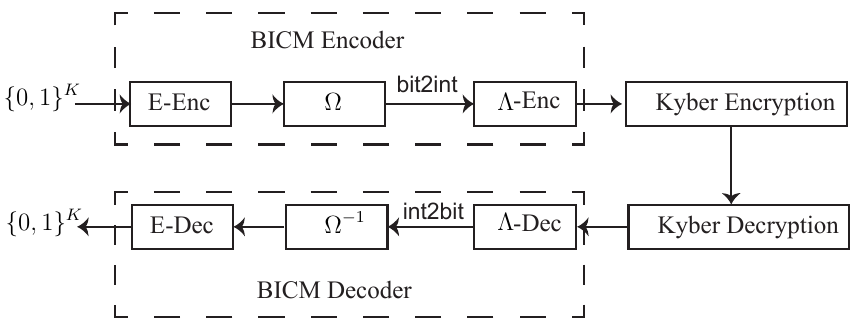} \vspace{0mm} 
\caption{Block diagram of the proposed bit-interleaved coded modulation.}
\label{K_BICM}
\end{figure}

We describe the BICM encoder and decoder in Algorithms 6 and 7, where the encoder $\mathsf{BICM}$-$\mathsf{Enc}$ takes as the input a message $m\in\mathcal{M}_{2,K}$ and outputs $v$ (i.e., the second part of the ciphertext), and where the decoder $\mathsf{BICM}$-$\mathsf{Dec}$ takes as input a ciphertext $(\bu,v)$ and outputs a message.

\vspace{-3mm}
\begin{algorithm}[H]
\caption{$\mathsf{BICM}$-$\mathsf{Enc}(m\in\mathcal{M}_{2,K}$): BICM encoder \\//Replace Step 5 in \textbf{Algorithm 2}}
\label{BICM_enc}
\begin{algorithmic}[1]

    \State
    $[b_{\ell,1}, \cdots, b_{\ell,\kappa}]^T:= \mathsf{E}$-$\mathsf{Enc}(m)$

    %\State
	%$\Omega \coloneqq\mathsf{Sam}(\rho)$

    \State
    $[\hat{b}_{\ell,1}, \cdots, \hat{b}_{\ell,\kappa}]^T:=\Omega([b_{\ell,1}, \cdots, b_{\ell,\kappa}]^T)$

    \State
    $m_{\ell,i}:=\mathsf{bit2int}(\hat{b}_{\ell,i})$, $i=1,\ldots, \kappa$

    \State
    $x_{\ell,i}:=\hat{\mathbf{B}}\cdot m_{\ell,i} \bmod p$, $i=1,\ldots, \kappa$

    \State
    $x:=[x_{\ell,1},\ldots, x_{\ell,\kappa}]^T$

    \State
   \Return  $v:= \mathsf{Compress}_{q}( \mathbf{t}^{T}\mathbf{r}+e_{2}+\lfloor q/p\rceil \cdot x,d_{v}) $

\end{algorithmic}
\end{algorithm}

\vspace{-3mm}

\begin{algorithm}[H]
\caption{$\mathsf{BICM}$-$\mathsf{Dec}(\bu,v)$: BICM decoder \\//Replace Step 3 in \textbf{Algorithm 3}}
\label{BICM_dec}
\begin{algorithmic}[1]

    \State
    $y:=v -\mathbf{s}^{T}\cdot\mathbf{u} = [y_{\ell,1},\ldots, y_{\ell,\kappa}]^T$, where $y_{\ell,i}$ is the $i^{\text{th}}$ block of coefficients in $y$

    \State
    ${m}_{\ell,i} :=\mathsf{HS-CVP}(y_{\ell,i}/\lfloor q/p\rceil, \mathcal{L}(\mathbf{B})),  i=1, \ldots, \kappa$

    \State
    $\hat{b}_{\ell,i}:=\mathsf{int2bit}(m_{\ell,i}), i=1,\ldots, \kappa$

    %\State
	%$\Omega \coloneqq\mathsf{Sam}(\rho)$

     \State
    $[{b}_{\ell,1}, \cdots, {b}_{\ell,\kappa}]^T:= \Omega^{-1}([\hat{b}_{\ell,1}, \cdots, \hat{b}_{\ell,\kappa}]^T)$

    \State
    \Return $m:=\mathsf{E}$-$\mathsf{Dec}([b_{\ell,1}, \cdots, b_{\ell,\kappa}]^T$)

\end{algorithmic}
\end{algorithm}

\vspace{-3mm}

\subsection{CER and DFR Reduction}
We now explain how to calculate the DFR for the proposed BICM approach. Let $\text{BER}_i \triangleq \Pr(\hat{b}_i \neq b_i)$ be the \emph{raw} bit error rate for the $i^{\text{th}}$ bit, for $i=1,\ldots, N$. Since the interleaver $\Omega$ spreads bursts of errors overall $N$ bits in $\kappa$ blocks, we can assume
\begin{equation}
\text{BER}_1 = \text{BER}_2 = \cdots =\text{BER}_N = \text{BER},
\end{equation}
and the coded bits go through essentially independent channels. It means that $\delta=1-(1-\text{BER})^N \approx N\cdot\text{BER}$. Therefore, the DFR with BICM, denoted as $\delta_c$, can be computed by
\begin{equation}
\delta_{c} =\sum_{j=t+1}^N
\begin{pmatrix}
N \\j
\end{pmatrix}
\text{BER}^j(1-\text{BER})^{N-j} 
<
\sum_{j=t+1}^N\begin{pmatrix}
N \\j
\end{pmatrix}
\text{BER}^j
\approx  \begin{pmatrix}
N \\t+1
\end{pmatrix}
\left(\dfrac{\delta}{N}\right)^{t+1}, \label{DFR_BICM}
\end{equation}
where $t$ is given in Definition 7, $N$ is given in (\ref{num_ebits}), and $\delta$ is given in (\ref{DFR_L}).

Considering the value of $\delta$ in Table \ref{KYBER_LC}, we see that the proposed BICM scheme can significantly reduce the DFR of Kyber for large $t$. This advantage allows us to reduce the ciphertext size by further compressing the first part of the ciphertext, i.e., $\bu$, since we can afford larger decoding noise. Let $\hat{d}_u$ be the compression parameter for $\bu$ used in the BICM approach. The CER reduction ratio for Kyber in (\ref{CER_R}) can then be calculated as:
\begin{equation}
\text{CER-R}=1-\dfrac{kn\hat{d}_u+nd_v}{knd_u+nd_v}\leq 1-\dfrac{\hat{d}_u}{d_u},
\end{equation}
where the values of $(n,k, d_u, d_v)$ are given in Table \ref{Kyber_Par}. Note that the value of $\delta_{c}$ in (\ref{DFR_BICM}) will be calculated based on $\hat{d}_u$.
\begin{example} For BW16 lattice with $p=4$, we construct a primitive narrow-sense BCH$(511,448)$ code and then shorten it to $(320,257)$ by removing the $191$ bits of padded $0$s \cite{ecclin2004}. This BCH code generates the codeword out of $256$ secret key bits, $63$ redundancy bits and $1$ padding bit. It can correct $t=7$ errors. Table \ref{BICM_Kyber} compares $(\delta_c,\text{CER-R})$ for difference $\hat{d}_u$. We observe that the BICM approach can significantly reduce DFR and CER of Kyber.
\end{example}

\begin{table}[th]
\centering
%\renewcommand{\arraystretch}{1.3} % adds row cushion
%\begin{threeparttable}
\caption{BCH$(320,257)$-BW16 for KYBER}
\label{BICM_Kyber}\centering
\begin{tabular}{|c|c|c|}
\hline
& $\hat{d}_{u}=9$ & $\hat{d}_{u}=8$ \\ \hline
$\var(\psi _{d_{u}})$ & $3.8$ & $14.1$ \\ \hline
$(\delta _{c},$CER-R$)$: KYBER512 & $(2^{-623},8.33\%)$ & $(2^{-194},16.67\%)$
\\ \hline
$(\delta _{c},$CER-R$)$: KYBER768 & $(2^{-683},8.82\%)$ & $(2^{-202},17.65\%)$
\\ \hline
$(\delta _{c},$CER-R$)$: KYBER1024 & $(2^{-791},16.33\%)$ & $%
(2^{-213},24.49\%)$ \\ \hline
\end{tabular}
%\tablefootnote{asdfafds}
%\end{threeparttable}
\end{table}

\subsection{Security}
As shown in Table \ref{BICM_Kyber}, the value of $\var(\psi _{d_{u}})$ increases as the value of $\hat{d}_{u}$ decreases. According to \cite{Kyber2021}, adding more rounding noise will make the M-LWE problem harder. Moreover, the proposed BICM scheme is resistant to timing attacks, as it uses constant-time BCH and lattice decoders.

\section{Conclusion}
Our analysis has shown that powerful error-correcting codes within the lattice encoder can lead to a significant improvement of important performance parameters of Kyber, such as security level, decryption failure rate, and communication cost. The lattice encoder can be viewed as a form of coded modulation, which ensures denser packing and lower raw bit error rate, while classical codes, e.g., BCH, can be used to get a high error-correcting capability. The proposed BICM approach, i.e., BCH codes through lattice modulation, combines their advantages to achieve a quasi-error-free key exchange with a high error-correcting capability. Through the use of BICM, we
have obtained enhanced parameter sets for Kyber, offering
higher security levels, smaller decryption failure rates, and smaller communication costs at the same time.

\section*{Data Availability Statement}
We do not analyse or generate any datasets, because our work proceeds within a theoretical and mathematical approach. One can obtain the relevant materials from the references below.

\bibliography{IEEEabrv,LIUBIB}

%% BioMed_Central_Bib_Style_v1.01

\begin{thebibliography}{23}
% BibTex style file: bmc-mathphys.bst (version 2.1), 2014-07-24
\ifx \bisbn   \undefined \def \bisbn  #1{ISBN #1}\fi
\ifx \binits  \undefined \def \binits#1{#1}\fi
\ifx \bauthor  \undefined \def \bauthor#1{#1}\fi
\ifx \batitle  \undefined \def \batitle#1{#1}\fi
\ifx \bjtitle  \undefined \def \bjtitle#1{#1}\fi
\ifx \bvolume  \undefined \def \bvolume#1{\textbf{#1}}\fi
\ifx \byear  \undefined \def \byear#1{#1}\fi
\ifx \bissue  \undefined \def \bissue#1{#1}\fi
\ifx \bfpage  \undefined \def \bfpage#1{#1}\fi
\ifx \blpage  \undefined \def \blpage #1{#1}\fi
\ifx \burl  \undefined \def \burl#1{\textsf{#1}}\fi
\ifx \doiurl  \undefined \def \doiurl#1{\url{https://doi.org/#1}}\fi
\ifx \betal  \undefined \def \betal{\textit{et al.}}\fi
\ifx \binstitute  \undefined \def \binstitute#1{#1}\fi
\ifx \binstitutionaled  \undefined \def \binstitutionaled#1{#1}\fi
\ifx \bctitle  \undefined \def \bctitle#1{#1}\fi
\ifx \beditor  \undefined \def \beditor#1{#1}\fi
\ifx \bpublisher  \undefined \def \bpublisher#1{#1}\fi
\ifx \bbtitle  \undefined \def \bbtitle#1{#1}\fi
\ifx \bedition  \undefined \def \bedition#1{#1}\fi
\ifx \bseriesno  \undefined \def \bseriesno#1{#1}\fi
\ifx \blocation  \undefined \def \blocation#1{#1}\fi
\ifx \bsertitle  \undefined \def \bsertitle#1{#1}\fi
\ifx \bsnm \undefined \def \bsnm#1{#1}\fi
\ifx \bsuffix \undefined \def \bsuffix#1{#1}\fi
\ifx \bparticle \undefined \def \bparticle#1{#1}\fi
\ifx \barticle \undefined \def \barticle#1{#1}\fi
\bibcommenthead
\ifx \bconfdate \undefined \def \bconfdate #1{#1}\fi
\ifx \botherref \undefined \def \botherref #1{#1}\fi
\ifx \url \undefined \def \url#1{\textsf{#1}}\fi
\ifx \bchapter \undefined \def \bchapter#1{#1}\fi
\ifx \bbook \undefined \def \bbook#1{#1}\fi
\ifx \bcomment \undefined \def \bcomment#1{#1}\fi
\ifx \oauthor \undefined \def \oauthor#1{#1}\fi
\ifx \citeauthoryear \undefined \def \citeauthoryear#1{#1}\fi
\ifx \endbibitem  \undefined \def \endbibitem {}\fi
\ifx \bconflocation  \undefined \def \bconflocation#1{#1}\fi
\ifx \arxivurl  \undefined \def \arxivurl#1{\textsf{#1}}\fi
\csname PreBibitemsHook\endcsname

%%% 1
\bibitem[\protect\citeauthoryear{{The NIST PQC Team}}{2022}]{NISTpqc2022}
\begin{botherref}
\oauthor{\bsnm{{The NIST PQC Team}}}:
{PQC} standardization process: Announcing four candidates to be standardized,
  plus fourth round candidates
(2022).
\url{https://csrc.nist.gov/News/2022/pqc-candidates-to-be-standardized-and-round-4}
\end{botherref}
\endbibitem

%%% 2
\bibitem[\protect\citeauthoryear{{National Institute of Standards and
  Technology}}{2023}]{NISTpqcdraft2023}
\begin{barticle}
\bauthor{\bsnm{{National Institute of Standards and Technology}}}:
\batitle{{Module-Lattice-based Key Encapsulation Mechanism Standard}}.
\bjtitle{Federal Information Processing Standards Publication (FIPS) NIST FIPS
  203 ipd.}
(\byear{2023})
\doiurl{10.6028/NIST.FIPS.203.ipd}
\end{barticle}
\endbibitem

%%% 3
\bibitem[\protect\citeauthoryear{Bos et~al.}{2018}]{Kyber2018}
\begin{bchapter}
\bauthor{\bsnm{Bos}, \binits{J.}},
\bauthor{\bsnm{Ducas}, \binits{L.}},
\bauthor{\bsnm{Kiltz}, \binits{E.}},
\bauthor{\bsnm{Lepoint}, \binits{T.}},
\bauthor{\bsnm{Lyubashevsky}, \binits{V.}},
\bauthor{\bsnm{Schanck}, \binits{J.M.}},
\bauthor{\bsnm{Schwabe}, \binits{P.}},
\bauthor{\bsnm{Seiler}, \binits{G.}},
\bauthor{\bsnm{Stehl\'e}, \binits{D.}}:
\bctitle{{CRYSTALS - Kyber}: A {CCA-Secure Module-Lattice-Based KEM}}.
In: \bbtitle{2018 IEEE European Symposium on Security and Privacy (EuroS\&P)},
pp. \bfpage{353}--\blpage{367}
(\byear{2018}).
\doiurl{10.1109/EuroSP.2018.00032}
\end{bchapter}
\endbibitem

%%% 4
\bibitem[\protect\citeauthoryear{D'Anvers et~al.}{2018}]{Sabar2018}
\begin{bchapter}
\bauthor{\bsnm{D'Anvers}, \binits{J.-P.}},
\bauthor{\bsnm{Karmakar}, \binits{A.}},
\bauthor{\bsnm{Sinha~Roy}, \binits{S.}},
\bauthor{\bsnm{Vercauteren}, \binits{F.}}:
\bctitle{{Saber: Module-LWR Based Key Exchange, CPA-Secure Encryption and
  CCA-Secure KEM}}.
In: \beditor{\bsnm{Joux}, \binits{A.}},
\beditor{\bsnm{Nitaj}, \binits{A.}},
\beditor{\bsnm{Rachidi}, \binits{T.}} (eds.)
\bbtitle{Progress in Cryptology -- AFRICACRYPT 2018},
pp. \bfpage{282}--\blpage{305}.
\bpublisher{Springer},
\blocation{Cham}
(\byear{2018}).
\doiurl{10.1007/978-3-319-89339-6_16}
\end{bchapter}
\endbibitem

%%% 5
\bibitem[\protect\citeauthoryear{Alkim et~al.}{2021}]{FrodoKEM2021}
\begin{botherref}
\oauthor{\bsnm{Alkim}, \binits{E.}},
\oauthor{\bsnm{Bos}, \binits{J.W.}},
\oauthor{\bsnm{Ducas}, \binits{L.}},
\oauthor{\bsnm{Longa}, \binits{P.}},
\oauthor{\bsnm{Mironov}, \binits{I.}},
\oauthor{\bsnm{Naehrig}, \binits{M.}},
\oauthor{\bsnm{Nikolaenko}, \binits{V.}},
\oauthor{\bsnm{Peikert}, \binits{C.}},
\oauthor{\bsnm{Raghunathan}, \binits{A.}},
\oauthor{\bsnm{Stebila}, \binits{D.}}:
{FrodoKEM: Learning With Errors Key Encapsulation}.
{NIST Round 3 specification}
(2021).
\url{https://frodokem.org/}
\end{botherref}
\endbibitem

%%% 6
\bibitem[\protect\citeauthoryear{D'Anvers et~al.}{2019}]{DFRAttack2019}
\begin{bchapter}
\bauthor{\bsnm{D'Anvers}, \binits{J.-P.}},
\bauthor{\bsnm{Guo}, \binits{Q.}},
\bauthor{\bsnm{Johansson}, \binits{T.}},
\bauthor{\bsnm{Nilsson}, \binits{A.}},
\bauthor{\bsnm{Vercauteren}, \binits{F.}},
\bauthor{\bsnm{Verbauwhede}, \binits{I.}}:
\bctitle{{Decryption Failure Attacks on IND-CCA Secure Lattice-Based Schemes}}.
In: \beditor{\bsnm{Lin}, \binits{D.}},
\beditor{\bsnm{Sako}, \binits{K.}} (eds.)
\bbtitle{Public-Key Cryptography -- PKC 2019},
pp. \bfpage{565}--\blpage{598}.
\bpublisher{Springer},
\blocation{Cham}
(\byear{2019}).
\doiurl{10.1007/978-3-030-17259-6_19}
\end{bchapter}
\endbibitem

%%% 7
\bibitem[\protect\citeauthoryear{Avanzi et~al.}{2021}]{Kyber2021}
\begin{botherref}
\oauthor{\bsnm{Avanzi}, \binits{R.}},
\oauthor{\bsnm{Bos}, \binits{J.}},
\oauthor{\bsnm{Ducas}, \binits{L.}},
\oauthor{\bsnm{Kiltz}, \binits{E.}},
\oauthor{\bsnm{Lepoint}, \binits{T.}},
\oauthor{\bsnm{Lyubashevsky}, \binits{V.}},
\oauthor{\bsnm{Schanck}, \binits{J.}},
\oauthor{\bsnm{Schwabe}, \binits{P.}},
\oauthor{\bsnm{Seiler}, \binits{G.}},
\oauthor{\bsnm{Stehl\'e}, \binits{D.}}:
Algorithm specifications and supporting documentation (version 3.02).
Tech. rep., Submission to the NIST post-quantum project
(2021).
\url{https://pq-crystals.org/kyber/resources.shtml}
\end{botherref}
\endbibitem

%%% 8
\bibitem[\protect\citeauthoryear{D'Anvers and Batsleer}{2022}]{DFRattack2022}
\begin{bchapter}
\bauthor{\bsnm{D'Anvers}, \binits{J.-P.}},
\bauthor{\bsnm{Batsleer}, \binits{S.}}:
\bctitle{{Multitarget Decryption Failure Attacks and Their Application to Saber
  and Kyber}}.
In: \beditor{\bsnm{Hanaoka}, \binits{G.}},
\beditor{\bsnm{Shikata}, \binits{J.}},
\beditor{\bsnm{Watanabe}, \binits{Y.}} (eds.)
\bbtitle{Public-Key Cryptography -- PKC 2022},
pp. \bfpage{3}--\blpage{33}.
\bpublisher{Springer},
\blocation{Cham}
(\byear{2022}).
\doiurl{10.1007/978-3-030-97121-2_1}
\end{bchapter}
\endbibitem

%%% 9
\bibitem[\protect\citeauthoryear{Fritzmann et~al.}{2019}]{NewhopeECC2018}
\begin{bchapter}
\bauthor{\bsnm{Fritzmann}, \binits{T.}},
\bauthor{\bsnm{P{\"o}ppelmann}, \binits{T.}},
\bauthor{\bsnm{Sepulveda}, \binits{J.}}:
\bctitle{Analysis of error-correcting codes for lattice-based key exchange}.
In: \beditor{\bsnm{Cid}, \binits{C.}},
\beditor{\bsnm{Jacobson~Jr.}, \binits{M.J.}} (eds.)
\bbtitle{Selected Areas in Cryptography -- SAC 2018},
pp. \bfpage{369}--\blpage{390}.
\bpublisher{Springer},
\blocation{Cham}
(\byear{2019}).
\doiurl{10.1007/978-3-030-10970-7_17}
\end{bchapter}
\endbibitem

%%% 10
\bibitem[\protect\citeauthoryear{Papadopoulos and Wang}{2023}]{Kyberpolar2022}
\begin{botherref}
\oauthor{\bsnm{Papadopoulos}, \binits{I.}},
\oauthor{\bsnm{Wang}, \binits{J.}}:
Polar codes for {Module-LWE} public key encryption: The case of {Kyber}.
Cryptography
\textbf{7}(1)
(2023)
\doiurl{10.3390/cryptography7010002}
\end{botherref}
\endbibitem

%%% 11
\bibitem[\protect\citeauthoryear{Lyu et~al.}{2023}]{FrodoCong2022}
\begin{bchapter}
\bauthor{\bsnm{Lyu}, \binits{S.}},
\bauthor{\bsnm{Liu}, \binits{L.}},
\bauthor{\bsnm{Ling}, \binits{C.}},
\bauthor{\bsnm{Lai}, \binits{J.}},
\bauthor{\bsnm{Chen}, \binits{H.}}:
\bctitle{{Lattice Codes for Lattice-Based PKE}}.
In: \bbtitle{Des. Codes Cryptogr.}
(\byear{2023}).
\burl{https://doi.org/10.1007/s10623-023-01321-6}
\end{bchapter}
\endbibitem

%%% 12
\bibitem[\protect\citeauthoryear{D'Anvers et~al.}{2019}]{ECCTimingAttack2019}
\begin{bchapter}
\bauthor{\bsnm{D'Anvers}, \binits{J.-P.}},
\bauthor{\bsnm{Tiepelt}, \binits{M.}},
\bauthor{\bsnm{Vercauteren}, \binits{F.}},
\bauthor{\bsnm{Verbauwhede}, \binits{I.}}:
\bctitle{{Timing Attacks on Error Correcting Codes in Post-Quantum Schemes}}.
In: \bbtitle{Proceedings of ACM Workshop on Theory of Implementation Security
  Workshop}.
\bsertitle{TIS'19},
pp. \bfpage{2}--\blpage{9}.
\bpublisher{Association for Computing Machinery},
\blocation{New York, NY, USA}
(\byear{2019}).
\doiurl{10.1145/3338467.3358948}
\end{bchapter}
\endbibitem

%%% 13
\bibitem[\protect\citeauthoryear{Walters and Roy}{2020}]{constantBCH2020}
\begin{bchapter}
\bauthor{\bsnm{Walters}, \binits{M.}},
\bauthor{\bsnm{Roy}, \binits{S.S.}}:
\bctitle{Constant-time {BCH} error-correcting code}.
In: \bbtitle{2020 IEEE International Symposium on Circuits and Systems
  (ISCAS)},
pp. \bfpage{1}--\blpage{5}
(\byear{2020}).
\doiurl{10.1109/ISCAS45731.2020.9180846}
\end{bchapter}
\endbibitem

%%% 14
\bibitem[\protect\citeauthoryear{van Poppelen}{2016}]{constanttimeLeech2016}
\begin{botherref}
\oauthor{\bsnm{Poppelen}, \binits{A.}}:
{Cryptographic decoding of the Leech lattice}.
Cryptology ePrint Archive, Paper 2016/1050
(2016).
\url{https://eprint.iacr.org/2016/1050}
\end{botherref}
\endbibitem

%%% 15
\bibitem[\protect\citeauthoryear{Gupta et~al.}{2021}]{KyberGPU2021}
\begin{barticle}
\bauthor{\bsnm{Gupta}, \binits{N.}},
\bauthor{\bsnm{Jati}, \binits{A.}},
\bauthor{\bsnm{Chauhan}, \binits{A.K.}},
\bauthor{\bsnm{Chattopadhyay}, \binits{A.}}:
\batitle{{PQC Acceleration Using GPUs: FrodoKEM, NewHope, and Kyber}}.
\bjtitle{IEEE Transactions on Parallel and Distributed Systems}
\bvolume{32}(\bissue{3}),
\bfpage{575}--\blpage{586}
(\byear{2021})
\doiurl{10.1109/TPDS.2020.3025691}
\end{barticle}
\endbibitem

%%% 16
\bibitem[\protect\citeauthoryear{Huang et~al.}{2020}]{KyberHard2020}
\begin{barticle}
\bauthor{\bsnm{Huang}, \binits{Y.}},
\bauthor{\bsnm{Huang}, \binits{M.}},
\bauthor{\bsnm{Lei}, \binits{Z.}},
\bauthor{\bsnm{Wu}, \binits{J.}}:
\batitle{{A pure hardware implementation of CRYSTALS-KYBER PQC algorithm
  through resource reuse}}.
\bjtitle{IEICE Electron. Express}
\bvolume{17},
\bfpage{20200234}
(\byear{2020})
\doiurl{10.1587/elex.17.20200234}
\end{barticle}
\endbibitem

%%% 17
\bibitem[\protect\citeauthoryear{Costa et~al.}{2022}]{KyberSmartMeter2022}
\begin{barticle}
\bauthor{\bsnm{Costa}, \binits{V.L.R.D.}},
\bauthor{\bsnm{Camponogara}, \binits{{\^A}.}},
\bauthor{\bsnm{L\'opez}, \binits{J.}},
\bauthor{\bsnm{Ribeiro}, \binits{M.V.}}:
\batitle{{The Feasibility of the CRYSTALS-Kyber Scheme for Smart Metering
  Systems}}.
\bjtitle{IEEE Access}
\bvolume{10},
\bfpage{131303}--\blpage{131317}
(\byear{2022})
\doiurl{10.1109/ACCESS.2022.3229521}
\end{barticle}
\endbibitem

%%% 18
\bibitem[\protect\citeauthoryear{Costache et~al.}{2022}]{CLTRLWE2022}
\begin{botherref}
\oauthor{\bsnm{Costache}, \binits{A.}},
\oauthor{\bsnm{Curtis}, \binits{B.R.}},
\oauthor{\bsnm{Hales}, \binits{E.}},
\oauthor{\bsnm{Murphy}, \binits{S.}},
\oauthor{\bsnm{Ogilvie}, \binits{T.}},
\oauthor{\bsnm{Player}, \binits{R.}}:
On the precision loss in approximate homomorphic encryption.
Cryptology ePrint Archive, Paper 2022/162
(2022).
\url{https://eprint.iacr.org/2022/162}
\end{botherref}
\endbibitem

%%% 19
\bibitem[\protect\citeauthoryear{Ding et~al.}{2022}]{KyberNoise2022}
\begin{bchapter}
\bauthor{\bsnm{Ding}, \binits{X.}},
\bauthor{\bsnm{Esgin}, \binits{M.F.}},
\bauthor{\bsnm{Sakzad}, \binits{A.}},
\bauthor{\bsnm{Steinfeld}, \binits{R.}}:
\bctitle{An injectivity analysis of {Crystals-Kyber} and implications on
  quantum security}.
In: \bbtitle{Information Security and Privacy},
pp. \bfpage{332}--\blpage{351}.
\bpublisher{Springer},
\blocation{Cham}
(\byear{2022}).
\doiurl{10.1007/978-3-031-22301-3_17}
\end{bchapter}
\endbibitem

%%% 20
\bibitem[\protect\citeauthoryear{Conway and Sloane}{1999}]{BK:Conway93}
\begin{bbook}
\bauthor{\bsnm{Conway}, \binits{J.H.}},
\bauthor{\bsnm{Sloane}, \binits{N.J.A.}}:
\bbtitle{Sphere Packings, Lattices, and Groups},
\bedition{3}rd edn.
\bpublisher{Springer},
\blocation{New York}
(\byear{1999}).
\doiurl{10.1007/978-1-4757-6568-7}
\end{bbook}
\endbibitem

%%% 21
\bibitem[\protect\citeauthoryear{Vardy and Be'ery}{1993}]{VD93Leech}
\begin{barticle}
\bauthor{\bsnm{Vardy}, \binits{A.}},
\bauthor{\bsnm{Be'ery}, \binits{Y.}}:
\batitle{Maximum likelihood decoding of the {Leech} lattice}.
\bjtitle{IEEE Transactions on Information Theory}
\bvolume{39}(\bissue{4}),
\bfpage{1435}--\blpage{1444}
(\byear{1993})
\doiurl{10.1109/18.243466}
\end{barticle}
\endbibitem

%%% 22
\bibitem[\protect\citeauthoryear{Caire et~al.}{1998}]{BICM1998}
\begin{barticle}
\bauthor{\bsnm{Caire}, \binits{G.}},
\bauthor{\bsnm{Taricco}, \binits{G.}},
\bauthor{\bsnm{Biglieri}, \binits{E.}}:
\batitle{{Bit-interleaved coded modulation}}.
\bjtitle{{IEEE} Trans. Inf. Theory}
\bvolume{44}(\bissue{3}),
\bfpage{927}--\blpage{946}
(\byear{1998})
\doiurl{10.1109/18.669123}
\end{barticle}
\endbibitem

%%% 23
\bibitem[\protect\citeauthoryear{Lin and Costello}{2004}]{ecclin2004}
\begin{bbook}
\bauthor{\bsnm{Lin}, \binits{S.}},
\bauthor{\bsnm{Costello}, \binits{D.J.}}:
\bbtitle{Error Control Coding, Second Edition}.
\bpublisher{Prentice-Hall, Inc.},
\blocation{USA}
(\byear{2004})
\end{bbook}
\endbibitem

\end{thebibliography}
%\bibliography{sn-bibliography}% common bib file
%% if required, the content of .bbl file can be included here once bbl is generated
%%\input sn-article.bbl

\end{document}